\documentclass[review,12pt]{elsarticle}

\usepackage{graphicx}
\usepackage{multirow}
\usepackage{amsmath,amssymb,amsfonts}
\usepackage{amsthm}
\usepackage{mathrsfs}
\usepackage[title]{appendix}
\usepackage{xcolor}
\usepackage{textcomp}
\usepackage{manyfoot}
\usepackage{booktabs}
\usepackage{algorithm}
\usepackage{algorithmicx}
\usepackage{algpseudocode}
\usepackage{listings}
\usepackage{tikz}
\usepackage{mathtools}
\usepackage{hyperref}
\usepackage{csquotes}
\usepackage{utfsym, adforn}
\usepackage{wrapfig}
\usepackage[utf8]{inputenc}
\usepackage{subcaption}
\usepackage{lineno}
\usepackage{booktabs, stfloats, tabularx}
\usepackage[normalem]{ulem}

\DeclareMathOperator*{\sinc}{sinc}

\theoremstyle{plain}
\newtheorem{theorem}{Theorem}
\newtheorem{proposition}[theorem]{Proposition}
\newtheorem{lemma}{Lemma}[section]

\theoremstyle{definition}
\newtheorem{definition}{Definition}

\theoremstyle{remark}
\newtheorem{remark}{Remark}

\journal{Signal Processing}

\begin{document}

\begin{frontmatter}

\title{Sampling with Autoregressive Clock Jitter: Infinite-Dimensional Stability Failure and Finite-Dimensional Conditioning}
\tnotetext[cc]{\textcopyright~2026. This manuscript version is made available under the CC-BY-NC-ND 4.0 license \url{https://creativecommons.org/licenses/by-nc-nd/4.0/}}

\author[1]{Daniele Gerosa\corref{cor1}}
\ead{danieleg@chalmers.se}
\cortext[cor1]{Corresponding author}

\author[1]{Thomas Eriksson}
\ead{thomase@chalmers.se}

\affiliation[1]{organization={Communication, Antennas and Optical Networks, Electrical Engineering Dept., Chalmers University of Technology},
    addressline={Chalmersplatsen 1},
    city={G\"oteborg},
    postcode={41296},
    country={Sweden}}

\begin{abstract}
Motivated by recent developments in stochastic modeling of clock jitter in Analog-to-Digital Converters (ADCs) as autoregressive processes of order one (AR(1)), we study the density and stability properties of AR(1)-jittered sampling sets for Paley-Wiener signals. While the stability of sampling sets under i.i.d. perturbations of the integer lattice is classical, no analogous analysis was previously available for the correlated jitter observed in practice. We show that, despite having the correct asymptotic density both on average and almost surely, such sets almost surely fail to be stable sampling sets. We complement this negative result with a finite-dimensional analysis, providing explicit non-asymptotic conditions on the number of samples \( 2N+1 \), the jitter standard deviation \( \sigma_\xi \) and the confidence level \( 1 - \eta \) under which the smallest singular value of the \( (2N+1) \times (2N+1) \) jittered sinc matrix exceeds a prescribed threshold \( \epsilon \) with probability at least \( 1-\eta \), so that the matrix is well-conditioned with high probability.
\end{abstract}


\end{frontmatter}

\section{Introduction}\label{sec:intro}
In signal processing, an analog waveform is typically modeled as a function of a continuous time parameter \(t \in \mathbb{R}\). Such functions, also called \emph{analog signals}, can be deterministic or random; in the latter case, we usually speak of continuous-time stochastic processes. Analog-to-digital converters (ADCs) are hardware systems that convert an analog signal into a digital representation, and one of the basic steps in this operation is \emph{sampling}. Samples may be described as instantaneous point evaluations of \(s\) at \(t_n\), i.e.\ \(s(t_n)\), or as averages over the sampling interval \(T_s > 0\), i.e.\ \( \frac{1}{T_s} \int_{nT_s}^{(n+1)T_s} s(t) \, dt  \) \cite{vetterlietal}. In this work, sampling will always be understood as point evaluation.

An analog signal \(s\) whose frequency spectrum is bounded can be fully reconstructed from its samples using the classical Shannon-Whittaker-Nyquist-Kotelnikov theorem (or some extension thereof), provided that \(1.\) all the infinitely many sample values \( \{s(t_n)\}_{n \in \mathbb{Z}} \) are accessible, \(2.\) the sampling interval \(T_s\) is sufficiently small, and \(3.\) the sampling instants \(t_n\) are either (a) exact multiples of \(T_s\), or (b) form an irregular sampling set satisfying suitable stability conditions; density conditions provide important necessary constraints, see e.g.\ Chapter 10 in \cite{higgins}. The theory concerning the reconstruction of bandlimited signals from samples taken on irregular time grids is by now mathematically mature, having been developed from multiple perspectives \cite{higgins,grochenig,margolis_eldar,strohmer_tanner,yen,seip_ulanovskii,bass_groechenig,bass_groechenig2}. As an example, the following theorem provides a counterpart to the Shannon-Whittaker-Nyquist-Kotelnikov theorem for irregular sampling grids:

\begin{theorem}[10.9-10.11 in \cite{higgins}] \label{thm:Kadec} Let \( \{ t_n \}_{n \in \mathbb{Z}} \subseteq \mathbb{R} \) satisfy Kadec's ``one quarter'' condition 

\begin{equation}
|n-t_n|\le D < 1/4
\tag*{[K]}
\end{equation}
for all \( n \in \mathbb{Z} \), and let \(s : \mathbb{R} \to \mathbb{R} \) such that \( \text{supp}(\mathcal{F}(s)) \subseteq [-1/2,1/2] \). Then \[ s(t) = \sum_{ n \in \mathbb{Z}} s(t_n) \frac{L(t)}{L'(t_n) (t-t_n)} \] 
in norm and globally uniformly over \( \mathbb{R} \), with \begin{equation} \label{eq:Lagr_infProd} L(t) = (t-t_0) \prod_{n=1}^\infty \left(1 - \frac{t}{t_n} \right) \left(1 - \frac{t}{t_{-n}} \right). \end{equation}
\end{theorem}

Although the product representation in \eqref{eq:Lagr_infProd} is generally not explicit, it reduces to closed-form expressions in some special cases \cite{higgins}. An extensive compendium of sampling theorems can be found in \cite{Zayed}, chapters \(2\) and \(3\). 

In concrete applications, hardware practitioners still tend to rely on the standard sinc interpolation series, mostly for reasons of complexity and computational cost \cite{tertinek_vogel}. Furthermore, in practice the actual sampling instants \(t_n\) are unknown and must be estimated. There is also significant evidence, both from measurements and from the properties of oscillators, that the sampling instants \(t_n\) are random in nature \cite{Zanchi_Samori} and should therefore be modeled as a stochastic process. Building on \cite{Towfic_Ting_Sayed,Chang_Lin_Wang_Lee_Shih}, the authors of \cite{gerosaetal} recently proposed to model the correlated component of the stochastic clock jitter caused by oscillator imperfections in ADCs using a two-sided discrete autoregressive process of order \(1\) (AR(1)), so that \( t_n = T_s n + \xi_n \), where 
\begin{equation} \label{eq:two-sidedAR1} \xi_n = \sum_{k=0}^\infty \varphi^k \epsilon_{n-k}  \end{equation} 
and \( \{ \epsilon_n\}_{n \in \mathbb{Z} } \) is a sequence of independent and identically distributed Gaussian random variables, \( \epsilon_n \sim \mathcal{N}(0, \sigma_\epsilon ^2) \) for all \(n \in \mathbb{Z} \), with \(0<\varphi<1\) and \(\varphi\) typically close to \(1\). The process \eqref{eq:two-sidedAR1} has power spectral density \cite{Kay} 
\[ \mathcal{S}_\xi (e^{i \omega}) = \frac{\sigma_\epsilon ^2}{1+\varphi^2 - 2 \varphi \cos(\omega)}, \quad \omega \in [-\pi,\pi], \]
so that requiring \( \varphi \approx 1 \) correctly encodes the fact that this jitter is ``slowly varying'' and that much of its energy is therefore concentrated at low frequencies.

Since \( \sum_{k \ge 0} \mathrm{var}(\varphi^k \epsilon_{n-k}) = \sigma_\epsilon ^2 \sum_{k \ge 0} \varphi^{2k} < \infty \), the Khintchine-Kolmogorov convergence theorem ensures that the infinite sum in \eqref{eq:two-sidedAR1} converges almost surely for every \(n\), and thus the \( \xi_n \) are well-defined. Lévy's continuity theorem then yields \( \xi_n \sim \mathcal{N}(0, \sigma_\epsilon ^2 / (1 - \varphi^2)) \) for all \(n \in \mathbb{Z} \). We remark that, since the \( \xi_n \) are Gaussian and hence unbounded, the jittered lattice \( t_n = T_s n + \xi_n \) almost surely fails to satisfy Kadec's condition [K], so that Theorem \ref{thm:Kadec} does not apply to the model \eqref{eq:two-sidedAR1}.

\subsection{Basic definitions}
We recall here basic definitions.
\begin{definition}[Fourier transform] The \emph{Fourier transform} \( \mathcal{F}(s)(\omega) \) of a function \( s: \mathbb{R} \to \mathbb{C} \) is formally defined as \[ \mathcal{F}(s)(\omega) \coloneqq  \int_{-\infty}^\infty s(t) e^{- 2 \pi i \omega t} \, dt.  \]
\end{definition}
\( \mathcal{F} \) defines a bounded linear operator \( L^1(\mathbb{R}) \to L^\infty(\mathbb{R}) \) that can be extended to all of \( L^2(\mathbb{R}) \) by standard density arguments. If \(B \subseteq \mathbb{R} \) is a union of finitely many measurable sets of finite and positive Lebesgue measure, we define 

\begin{definition}[Paley-Wiener space] The \emph{Paley-Wiener space} is the set \[ \text{PW}_{B} \coloneqq \{s \in C(\mathbb{R}) \cap L^2(\mathbb{R}) \, : \, \text{supp}(\mathcal{F}(s)) \subseteq \text{cl}(B)  \}. \]
\end{definition}
In sampling theory, one basic question is whether a signal can be recovered from its samples. For applications, however, uniqueness alone may not be enough: the recovery should also be stable with respect to measurement, numerical, or modeling errors. In other words, a signal with large \(L^2\)-energy should not be almost ``invisible'' on the sampling set. This robustness requirement is
encoded in the notion of \emph{stable sampling}.
\begin{definition}[Stable sampling] \label{eq:stable_sampling} A set \( \{t_n\} \subseteq \mathbb{R} \) is of \emph{stable sampling} for \( \mathrm{PW}_B \) if \[ s \in \text{PW}_B \Longrightarrow \|s\|_{L^2} \le K \|s(t_n)\|_{\ell^2} \]with \(K\) independent of \(s\). 
\end{definition}
Sampling stability for Paley-Wiener functions is connected to the notion of \emph{density}, at least through a theorem of H.~J.~Landau \cite{Landau}, which states that the density \eqref{eq:set_density} of a sampling set must be at least the Lebesgue measure of \(B\). In this note we adopt the following classical definition of density.

\begin{definition}[Sampling set density] The density \(D(t_n)\) of a sampling set \( \{ t_n \} \subseteq \mathbb{R} \) is defined as 
\begin{equation} \label{eq:set_density} D(t_n) \coloneqq \lim_{r \to \infty} \frac{\text{card}(\{t_n \, : \, t_n \in [-r,r] \})}{2 r} \end{equation}whenever the latter limit exists.
\end{definition}

\subsection{Symbols and notations}
We adopt the notation of using uppercase boldface letters \( \mathbf{A} \) for matrices and lowercase boldface letters \( \mathbf{a} \) for column vectors. Due to the structural problem setup, we indicate with \(  \mathbf{a}_N  \) vectors belonging to \( \mathbb{R}^{2N+1} \) indexed from \( -N \) to \(N\). Similarly for \( (2N +1) \times (2N +1) \) matrices, that will be indicated by \( \mathbf{A}_N \) or \( \mathbf{A}^{(N)} \) if ambiguities may arise. We use \( \sigma_j (\mathbf{A}^{(N)}) \) for the \(j-\)th singular value of \( \mathbf{A}^{(N)}\), and \( \odot \) for the Hadamard component-wise product. The operator norm \( \| \mathbf{A}^{(N)} \|_{\text{op}} \) of a matrix equals its largest singular value \( \sigma_{\max} (\mathbf{A}^{(N)}) \).

\subsection{Motivations}
The study of the properties of systems of complex exponentials \( \{ e^{i \lambda_n x} \, : \, \lambda_n \in \Lambda \} \), with \( \Lambda \subseteq \mathbb{R} \), is classical and the corresponding theory is well developed, cf. \cite{paley_wiener}. Roughly speaking, ``good'' properties of the \( \{ e^{i \lambda_n x} \} \) imply ``good'' properties of the set \( \Lambda \) seen as a sampling set for Paley-Wiener functions, and vice versa \cite{higgins}. Much less is known when the \( \lambda_n \) are assumed to be random variables, although a fair amount has been established in the i.i.d.\ case \cite{bass_groechenig, bass_groechenig2, ChistyakovLyubarskii, seip_ulanovskii}; when correlation enters the picture, the results are quite sparse. We are not aware of results addressing sampling stability for correlated random perturbations of the integer lattice; the present note is a modest attempt to fill this gap, taking as a starting point the autoregressive jitter model of \cite{gerosaetal}. The mathematical machinery used in \cite{bass_groechenig, bass_groechenig2, ChistyakovLyubarskii, seip_ulanovskii} does not generalize well to the correlated scenario: independence is used there to reduce the analysis to product measures and to invoke Borel-Cantelli-type arguments on independent events, neither of which is available for a correlated jitter process. We therefore rely on the ergodicity of the shift acting on the AR(1) path, which supplies the required almost sure statements without any independence assumption.
Concretely, we investigate whether the autoregressively jittered lattice \( t_n = n+\xi_n \) retains the stable sampling properties of Definition \ref{eq:stable_sampling} of the integer lattice. 

\subsection{Contributions}

The contributions of this manuscript may be summarized as follows:
\begin{itemize}
    \item We prove that the autoregressive jittered lattice \( t_n = n + \xi_n \), with \( \xi_n \) defined in \eqref{eq:two-sidedAR1}, has the correct critical density \eqref{eq:set_density} on average (Lemma \ref{lemma:exp_decay}) and almost surely (Lemma \ref{lemma:exp_decay_AR1});
    \item We prove that \( \{t_n\} \) is \emph{not} of stable sampling for the Paley-Wiener space (Theorem \ref{thm:ar_jitt_notstable});
    \item We prove that the corresponding \(\sinc\) matrices are well-conditioned with high probability, both for the full sampling matrix (Proposition \ref{prop:prob_lb_sigmamin}) and for its first-order Taylor surrogate (Proposition \ref{prop:linearized_prob_lb});
    \item We complement these results with a Monte Carlo study (Section \ref{sec:num_exp}) of the empirical distributions of the smallest singular values and condition numbers, and we numerically quantify the conservativeness of the bounds obtained above.
\end{itemize}

\newpage
\section{Critical density and global stability failure} 
As observed by Bass and Gröchenig in \cite{bass_groechenig}, random sampling sets may, on average, attain the critical (lower Beurling) density required by Landau's results and still fail to be stable. We show in this section that this is precisely what happens for the AR(1)-jittered lattice.

\begin{lemma} \label{lemma:exp_decay} Let \( \{ \eta_n \}_{n \in \mathbb{Z}} \) be a sequence of jointly Gaussian real-valued random variables with \(0\) mean and variance \( \sigma_\eta ^2  \). Define \( \theta_n \coloneqq n + \eta_n \). Then \[ \frac{\mathbb{E}[\emph{card}(\{\theta_n \, : \, \theta_n \in [-r,r] \})] }{2 r} = 1 + O(r^{-1}).  \]
\end{lemma}
\begin{proof}
Consider the nonnegative random variables \( \mathbf{1}_{ \{ |\theta_n| \le r \} }.  \) By the Tonelli theorem \( \mathbb{E}[ \sum_{n \in \mathbb{Z}  } \mathbf{1}_{ \{ |\theta_n| \le r \}} ] = \sum_{n \in \mathbb{Z}} \mathbb{E}[\mathbf{1}_{ \{ |\theta_n| \le r \}}] \) and \begin{equation} \label{eq:sumexp} \begin{split} \sum_{n \in \mathbb{Z}} \mathbb{E}[\mathbf{1}_{ \{ |\theta_n| \le r \}}] & = \frac{1}{\sqrt{2 \pi}\sigma_\eta } \sum_{n \in \mathbb{Z}} \int_{-r}^r e^{- \frac{(s - n)^2}{2 \sigma_\eta ^2}} \, ds \\ & = \frac{1}{\sqrt{2 \pi} \sigma_\eta} \int_{-r}^r \sum_{n \in \mathbb{Z}} e^{- \frac{(s - n)^2}{2 \sigma_\eta ^2}} \, ds \\& = \frac{1}{\sqrt{2 \pi} \sigma_\eta } \int_{-r}^r e^{- \frac{s^2}{2 \sigma_\eta ^2}} \boxed{ \sum_{n \in \mathbb{Z}} e^{\frac{-n^2 + 2sn}{2 \sigma_\eta ^2}}} \, ds. \end{split}  \end{equation}
The boxed term is the Jacobi's theta function \( \vartheta(z;\tau) \) with \( \tau = i/(2 \pi  \sigma_\eta ^2) \) and \( z = -is / (2  \sigma_\eta ^2) \). The whole integrand becomes then equal to \[ \begin{split}  & =  e^{-s^2 / (2 \sigma_\eta ^2)} \vartheta \left(- \frac{is}{2 \sigma_\eta ^2 }; \frac{i}{2 \pi \sigma_\eta^2} \right) \\ & = e^{-s^2 / (2 \sigma_\eta ^2)} \sqrt{2 \pi \sigma_\eta ^2} e^{s^2 / (2 \sigma_\eta^2)} \vartheta(\pi s; 2 \pi i \sigma_\eta ^2) \\ & = \sqrt{2 \pi} \sigma_\eta \vartheta(\pi s; 2 \pi i \sigma_\eta ^2)  \end{split}   \]by using modular properties of the Theta function. \eqref{eq:sumexp} becomes then \[ \begin{split} \int_{-r}^r \vartheta(\pi s; 2 \pi i \sigma_\eta ^2) \, ds & = \int_{-r}^r \left[1 + 2 \sum_{n=1}^\infty e^{-2 \pi^2 \sigma_\eta^2 n^2} \cos(2 \pi n s)  \right] \, ds \\ & = 2r + 4 \sum_{n=1}^\infty e^{-2 \pi^2 \sigma_\eta^2 n^2} \frac{\sin(2 \pi n r) }{2 \pi n}  \end{split} \]
from which we have that, as \(r \to \infty \), \[ \frac{\mathbb{E}[\text{card}(\{\theta_n \, : \, \theta_n \in [-r,r] \})] }{2 r} = 1 + O(r^{-1}).  \] 
\end{proof}

So far we have not exploited the autoregressive structure of our jitter. Indeed, Lemma \ref{lemma:exp_decay} assumes nothing more than Gaussianity of the \( \eta_n \). In the next lemma we show that, under a mild summability assumption on the autocorrelation function, the convergence to the correct density can be upgraded from expectation to almost sure convergence. 

\begin{lemma} \label{lemma:exp_decay_AR1}
Let \( \{ \xi_n \}_{n \in \mathbb{Z}} \) be a centered stationary Gaussian process with variance \( \sigma_\xi^2 \), and let
\[
\rho(k) \coloneqq \operatorname{Corr}(\xi_n,\xi_{n+k}), \qquad k \in \mathbb{Z},
\]
denote its autocorrelation function. Assume that
\[
\sum_{k \in \mathbb{Z}} |\rho(k)| < \infty.
\]
Define \( \theta_n \coloneqq n + \xi_n \). Then
\[
S(r) \coloneqq \frac{1}{2r} \sum_{n \in \mathbb{Z} } \mathbf{1}_{ \{ |\theta_n| \le r \} } \to 1
\]
a.s. as \( r \to \infty \).
\end{lemma}

\begin{proof}
We study \( \text{var}(S(r)) \) for a fixed \(r\). We have
\begin{equation}\label{eq:varsplit}
\begin{split}
\operatorname{var}(S(r))
&= \frac{1}{4r^2}
\left[
\sum_{n\in\mathbb{Z}} \operatorname{var}\!\bigl(\mathbf{1}_{\{|\theta_n|\le r\}}\bigr) + 2\sum_{m<n} \operatorname{cov}\!\bigl(\mathbf{1}_{\{|\theta_m|\le r\}},
\mathbf{1}_{\{|\theta_n|\le r\}}\bigr)
\right] \\
&= \frac{1}{4r^2}
\left[
\sum_{n\in\mathbb{Z}} p_{n,r}(1-p_{n,r}) + 2\sum_{m<n} \operatorname{cov}\!\bigl(\mathbf{1}_{\{|\theta_m|\le r\}},
\mathbf{1}_{\{|\theta_n|\le r\}}\bigr)
\right].
\end{split}
\end{equation}
with
\[
p_{n,r} = \frac{1}{\sqrt{2 \pi} \sigma_\xi  } \int_{-r}^r e^{ - \frac{(s-n)^2}{2 \sigma_\xi^2 } } \, ds.
\]
To deal with the term \( \sum_n p_{n,r}(1 - p_{n,r}) \) we modify the argument in \eqref{eq:sumexp}:

\[
\begin{split}
p_{n,r}(1 - p_{n,r})
&= \frac{1}{2 \pi \sigma_{\xi} ^2} \int_{-r}^r \int_{\mathbb{R} \setminus [-r,r] } e^{ -\frac{(s-n)^2}{2 \sigma_\xi ^2} } e^{ -\frac{(t-n)^2}{2 \sigma_\xi ^2} } \, ds \, dt \\
&= \frac{1}{2 \pi \sigma_{\xi} ^2} \int_{-r}^r \int_{\mathbb{R} \setminus [-r,r] } e^{ -\frac{(s-t)^2}{4 \sigma_\xi ^2} } e^{ -\frac{((t+s)/2-n)^2}{ \sigma_\xi ^2} } \, ds \, dt.
\end{split}
\]
and
\[
\sum_{n \in \mathbb{Z}} e^{-((t+s)/2-n)^2/\sigma_\xi^2} \le C(\sigma_\xi)
\]
because the infinite sum defines a continuous and \(1\)-periodic function. Therefore
\begin{equation} \label{eq:varest}
\sum_{n} p_{n,r}(1 - p_{n,r}) \le C_1(\sigma_\xi) \int_{-r}^r \int_{\mathbb{R} \setminus [-r,r] } e^{ -\frac{(s-t)^2}{4 \sigma_\xi ^2} } \, ds \, dt.
\end{equation}
To bound the right-hand side of \eqref{eq:varest} it is enough to observe that
\[
\int_{r}^{\infty}e^{ -(s-t)^2/(4 \sigma_\xi ^2) } \, ds \lesssim e^{ -(r-t)^2/(4 \sigma_\xi ^2) }
\]
(and similarly for \( \int_{-\infty}^{-r} \)) as a consequence of elementary bounds of the Error function, and hence \eqref{eq:varest} is bounded from above by a constant independent of \(r\).

We need now to estimate the series of covariances in \eqref{eq:varsplit}, and to do so we make use of Gebelein \cite{gebelein_orig} as well as the AM-GM inequalities:
\[
\begin{split}
| \operatorname{cov}(\mathbf{1}_{ \{ |\theta_m|  \le r \} },\mathbf{1}_{ \{ |\theta_n| \le r \} }) |
&\le |\rho(n-m)| \sqrt{ \operatorname{var}(\mathbf{1}_{ \{ |\theta_m| \le r \} }) \operatorname{var}(\mathbf{1}_{ \{ |\theta_n| \le r \} }) } \\
&\le |\rho(n-m)| \frac{\operatorname{var}(\mathbf{1}_{ \{ |\theta_m| \le r \} }) + \operatorname{var}(\mathbf{1}_{ \{ |\theta_n| \le r \} })}{2}.
\end{split}
\]
Thus
\[
\sum_{m<n} |\rho(n-m)| \operatorname{var}(\mathbf{1}_{ \{ |\theta_m| \le r \} })
= \sum_{m \in \mathbb{Z}}\operatorname{var}(\mathbf{1}_{ \{ |\theta_m| \le r \} }) \sum_{h=1}^{\infty} |\rho(h)| \le C_2(\sigma_\xi,\rho),
\]
and similarly for the series with \( \operatorname{var}(\mathbf{1}_{ \{ |\theta_n| \le r \} })\). We can finally conclude that
\begin{equation} \label{eq:varineq}
\operatorname{var}(S(r)) \le C_3( \sigma_\xi, \rho) / r^2,
\end{equation}
which together with Lemma \ref{lemma:exp_decay} implies that \( S(r) \to 1 \) in mean square, and hence in probability. Convergence almost everywhere requires a bit more work: we first show that \( S(n) \to 1 \) almost everywhere along integers, and then extend the result to arbitrary sequences \( r_k \to \infty \).

As a consequence of Chebyshev's inequality and \eqref{eq:varineq}, for a fixed \( \epsilon > 0 \) we have that
\[
\mathbb{P}[\{ |S(n) -  \mathbb{E}[S(n)] | > \epsilon\} ] \lesssim \frac{1}{\epsilon^2 n ^2}
\]
and thus
\[
\sum_{n \ge 1} \mathbb{P}[\{ |S(n) -  \mathbb{E}[S(n)] | > \epsilon\} ] < \infty.
\]
The Borel-Cantelli lemma then gives
\begin{equation} \label{eq:borelcantelli}
\mathbb{P} \left[ \bigcap_{n=1}^{\infty} \bigcup_{k=n}^\infty \{ |S(k) -  \mathbb{E}[S(k)] | > \epsilon\}    \right] = 0,
\end{equation}
which implies \( S(n) - \mathbb{E}[S(n)] \to 0 \) almost surely. Combined with Lemma \ref{lemma:exp_decay}, this yields \(  S(n) \to 1 \) almost surely.

To conclude the argument, pick a sequence \( r_k \to \infty \) and define the two integers \( m_{k}^{-} = \lfloor r_k \rfloor - 1 \) and \( m_k ^{+} = \lfloor r_k  \rfloor + 1 \), so that \( m_{k}^{-} \le r_k \le m_k ^{+} \) for all \(k\) (the \(m\)'s depend, of course, on the sequence \(r_k\)). Then
\[
S(r_k) \le \frac{1}{2 r_k} \sum_{n \in \mathbb{Z}} \mathbf{1}_{ \{ |\theta_n| \le m_k ^{+} \} } = \frac{m_k ^{+}}{r_k} S(m_k ^{+})
\]
and similarly
\[
S(r_k) \ge \frac{m_{k}^{-}}{r_k} S(m_{k}^{-}).
\]
Since by construction both \( m_{k}^{-} / r_k \) and \( m_k ^{+} / r_k \) tend to \(1\) and \( S(m_{k}^{-}) \), \( S(m_k ^{+}) \) converge to \(1\) a.s., then so does \( S(r_k) \).
\end{proof}

\begin{remark}
Strictly speaking, the sum \( \sum_{n \in \mathbb{Z}} \mathbf{1}_{ \{ |\theta_n| \le r \} } \), which appears in the proof of Lemma \ref{lemma:exp_decay} and equals \( 2 r S(r) \) in Lemma \ref{lemma:exp_decay_AR1}, and \( \mathrm{card}(\{\theta_n \, : \, \theta_n \in [-r,r] \}) \) are not exactly the same quantity; however, since joint Gaussianity implies that \( \bigcap_{n \ne m} \{ \theta_n \ne \theta_m \} \) has probability \(1\), the two agree almost surely for every \(r>0\).
\end{remark}

We are now in position to prove the main negative result of this section. In the following theorem we adapt the construction of Proposition 2.2 in \cite{bass_groechenig} to show that the stable sampling condition introduced above is violated by the autoregressively jittered lattice. The bandlimited functions used in the counterexample are precisely those of \cite{bass_groechenig}; the proof strategy, however, must be adapted to account for the fact that the random variables \(\xi_n\) are no longer independent.

\begin{theorem} \label{thm:ar_jitt_notstable}
Let \( \theta_n(\omega) \coloneqq n+\xi_n(\omega) \), where \(\{\xi_n\}_{n\in\mathbb Z}\)
is the two-sided Gaussian AR(1) process defined in \eqref{eq:two-sidedAR1}, with
\(0<\varphi<1\). Then with probability \(1\) there exists a sequence of nonzero bandlimited signals \( \{s_k^\omega\}_{k \in \mathbb{N} } \subseteq \mathrm{PW}_{[-1/2,1/2]}  \) such that \[ \sum_{n\in\mathbb Z} |s_k^\omega (\theta_n(\omega))|^2 \le \frac{\| s_k^\omega\|_2 ^2}{k} \quad \forall \, k.  \] 
Consequently, the jittered lattice is almost surely not a stable sampling set for \(\mathrm{PW}_{[-1/2,1/2]}\).
\end{theorem}
\begin{proof}
We work in the critical normalization \( B = [-1/2,1/2] \). The process \(\{\xi_n\}_{n\in\mathbb Z}\) is realized as a measurable map \( \boldsymbol{\xi} : (\Omega',\mathbb P) \to \mathbb{R}^{\mathbb{Z}} \) on an underlying probability space \((\Omega',\mathbb P)\), and we denote by \( \mu \coloneqq \mathbb{P}\circ\boldsymbol{\xi}^{-1} \) its law on \(\mathbb{R}^{\mathbb{Z}}\). All the events introduced below are sets in \(\mathbb{R}^{\mathbb{Z}}\); for such an event \(E\) we write \( \mathbb{P}(\boldsymbol{\xi}\in E)=\mu(E) \), and the finite-dimensional marginals of \(\mu\) are the corresponding Gaussian laws of the relevant coordinates of \(\boldsymbol{\xi}\).

The function \( S(x) \coloneqq \sin(\pi x/2)\,\Psi(x) \), with \( \Psi \) as in the proof of Proposition 2.2 in \cite{bass_groechenig}, is a nonzero element of \( \mathrm{PW}_{[-1/2,1/2]} \) satisfying \(S(2j)=0\) for all \(j\in\mathbb Z\), with \( |S(x)| \le a / (1 + |x|^2) \) for some positive constant \(a\) and all \( x \in \mathbb{R} \). 
For each fixed \(k\), we introduce the two events 
\[ A_{\delta_k, N_k} \coloneqq \{ x \in \mathbb{R}^{\mathbb{Z}} : |x_{2j} | \le \delta_k, \ |x_{2j+1} - 1 | \le \delta_k, \ |j| \le N_k   \}  \] 
and 
\[ B_{N_k} \coloneqq \{ x \in \mathbb{R}^{\mathbb{Z}} : |x_j| \le |j| / 4, \ j \le -(2N_k+1) \text{ or } j \ge 2N_k+2 \}.  \]
Since \(S\) is bounded, so is \(S'\), by e.g.\ Theorem 11.1.2 in \cite{boas_entire}. We then choose \( \delta_k ^2 = \|S\|_2 ^2 / (8 N_k k \|S'\|_\infty ^2) = C_1 / (k N_k) \). Using the autoregressive structure of \(\xi_n\), the random vector \( (\xi_{-2N_k}, \dots ,\xi_{2 N_k + 1}) \) has precision matrix \cite{fikioris} \begin{equation} \label{eq:inv_cov_AR1} \boldsymbol{\Sigma}^{-1} _{4N_k + 2} = \frac{1}{\sigma_\epsilon^2} \begin{pmatrix} 1 & -\varphi & 0 & \dots & 0\\ -\varphi & 1+\varphi^2 & -\varphi & \dots & 0 \\ \vdots & \ddots & \ddots & \dots & \vdots \\ 0 & 0 & -\varphi & 1+\varphi^2  & -\varphi \\ 0 & 0 & 0 & -\varphi  & 1   \end{pmatrix}   \end{equation} and \( \det \left(\boldsymbol{\Sigma}^{-1} _{4N_k + 2}\right) = (1-\varphi^2)/\sigma_\epsilon^{2(4N_k+2)} \); equivalently, the marginal of \(\mu\) on the coordinates \(-2N_k,\dots,2N_k+1\) has density \[  f_{4N_k+2} (\mathbf{x})  = \frac{\sqrt{1-\varphi^2}}{(\sqrt{2 \pi}\sigma_\epsilon)^{4N_k+2}} \exp \left(-\mathbf{x}^t\boldsymbol{\Sigma}^{-1} _{4N_k + 2}\mathbf{x} / 2  \right).    \]
By using \eqref{eq:inv_cov_AR1} we can write explicitly the quadratic form in the exponential, which gives us \[ \begin{split} \sigma_\epsilon^2 \mathbf{x}^t\boldsymbol{\Sigma}^{-1} _{4N_k + 2}\mathbf{x} &= x_1 ^2 + x_{4N_k+2} ^2 + (1 + \varphi^2) \sum_{i=2}^{4N_k+1} x_i ^2   -2 \varphi \sum_{i=1}^{4 N_k + 1} x_{i} x_{i+1}  \\ & \le  \left( \delta_k ^2 + (1+\delta_k)^2 \right.   + (1+\varphi^2)( 2N_k(1+\delta_k)^2 + 2 N_k \delta_k ^2) \\& + \left. 2 \varphi (4N_k+1) \delta_k (1+\delta_k)  \right) = \ell   \end{split} \] if \( \mathbf{x} \in (0,1,\dots,0,1) + [-\delta_k, \delta_k]^{4N_k + 2} \). Therefore \[ \begin{split} \mu(A_{\delta_k, N_k}) & \ge C_2 (\varphi) \left( \frac{\sqrt{2}\delta_k}{\sqrt{\pi} \sigma_\epsilon } \right)^{4N_k+2} e^{- \frac{\ell}{2\sigma_\epsilon ^2}  } \\& = C_2 (\varphi) \left( \frac{C_3}{ \sqrt{N_k} } \right)^{4N_k+2} e^{- \frac{1}{2\sigma_\epsilon ^2} \left( O(N_k) + O\left(N_k^{1/2} \right) + O(1) + O \left(N_k^{-1/2} \right) + O\left(N_k^{-1} \right)     \right). }   \end{split} \] In logarithmic scale, the latter becomes 
\[ \log(\mu(A_{\delta_k, N_k})) \ge C_4 - C_5 N_k \log(N_k) - C_6 N_k  \] 
for large \(N_k\), where the constants \(C_i\) are strictly positive and depend only on \( \varphi \), \( \sigma_\epsilon \), and \(k\). On the other hand, by the union bound,
\[ \begin{split} \mu(B_{N_k}^c) & \le \sum_{j \le -(2N_k+1)} \mu(|x_j|>|j|/4) + \sum_{j \ge 2N_k+2} \mu(|x_j|>|j|/4) \\ &\lesssim \sum_{j \ge 2N_k+1} \exp(- C_7 j^2) \lesssim C_8 e^{-C_9 N_k^2} \end{split}\]
for some other strictly positive constants \( C_7, C_8, C_9 \). Combining the two,
\begin{equation} \label{eq:pos_prob} \mu(A_{\delta_k, N_k} \cap B_{N_k}) \ge \mu(A_{\delta_k, N_k}) - \mu(B_{N_k} ^c),  \end{equation} 
and by the asymptotics in \(N_k\) the right-hand side is eventually strictly positive.
By choosing \(N_k\) large enough we can also ensure that
\begin{equation} \label{eq:tail}\sum_{|j| \ge 2 N_k } \frac{a ^2}{(1 + (3|j|/4)^2)^2  } \le \frac{\|S\|_2 ^2}{4k}. \end{equation} 
We henceforth fix \( N_k \) satisfying both \eqref{eq:tail} and \eqref{eq:pos_prob}.
Consider now the shift operator \( T : \mathbb{R}^{\mathbb{Z}} \to \mathbb{R}^{\mathbb{Z}} \) defined by \( (T x)_n = x_{n+1} \). The stationary Gaussian AR(1) process with \(0<\varphi<1\) is mixing by Theorem \(10\) in \cite{krishnapur}; equivalently, its law \( \mu = \mathbb{P} \circ \boldsymbol{\xi}^{-1} \) is mixing, hence ergodic, with respect to \(T\). It follows that the system \((\mathbb{R}^{\mathbb{Z}},\mu,U)\) with \(U\coloneqq T^2\) is also mixing, and thus ergodic. Set \( E_k \coloneqq A_{\delta_k, N_k} \cap B_{N_k} \subseteq \mathbb{R}^{\mathbb{Z}} \), so that \( \mu(E_k)>0 \) by \eqref{eq:pos_prob}, and define
\[
F_k \coloneqq \bigcup_{m \ge 1}U^{-m} E_k.
\]
By the characterization of ergodicity (cf. Theorem 1.5.(iii) in \cite{walters})\footnote{Theorem 1.5.(iii) in \cite{walters} requires the maps \(T\) and \(U\) to be measure-preserving, i.e. \( \mu(A) = \mu \circ T^{-1} (A) \) for all \( A \in \mathcal{B}(\mathbb{R}^\mathbb{Z}) \). This is true in our case but it requires a bit of a tedious proof, that we sketch here: for \(I = (n_1, \dots , n_d) \in \mathbb{Z}^d \) with \( n_1 < \dots < n_d \) consider the canonical projections \( \pi_I : \mathbb{R}^\mathbb{Z} \to \mathbb{R}^{|I|} \cong \mathbb{R}^d \) and the sets \( \mathcal{C}_d = \{ \pi_I^{-1} (A) \, : \, A \in \mathcal{B}(\mathbb{R}^d), \, |I|=d  \} \). By definition \( \sigma (\mathcal{C}_1) = \mathcal{B}(\mathbb{R}^\mathbb{Z}) \) and moreover \( \mathcal{C} \coloneqq \bigcup_{d \ge 1} \mathcal{C}_d \) is a \( \pi- \)system, thus it is enough to prove \( \mu(C) = \mu \circ T^{-1} (C) \) for all \( C \in \mathcal{C} \) by Dynkin's lemma (1.6(a) in \cite{williams}). The conclusion follows from the strict stationarity of \( \boldsymbol{\xi} \).} we have \( \mu(F_k)=1 \), and hence \( \mathbb{P}(\Omega_k) = 1 \) for \( \Omega_k \coloneqq \boldsymbol{\xi}^{-1}(F_k) \). Therefore, for every \(\omega \in \Omega_k\) there exists at least one \(\overline{m}\) with \( U^{\overline{m}}(\boldsymbol{\xi}(\omega)) \in E_k \), i.e.\ \( T^{2\overline{m}}(\boldsymbol{\xi}(\omega)) \in A_{\delta_k, N_k} \cap B_{N_k} \). Define \( s_k^\omega (x) \coloneqq S(x - 2\overline{m}) \). Since \( (T^{2\overline{m}}\boldsymbol{\xi}(\omega))_j = \xi_{j+2\overline{m}}(\omega) \), the membership \( T^{2\overline{m}}(\boldsymbol{\xi}(\omega))\in B_{N_k} \) means that for indices with \( j - 2\overline{m} \le -(2N_k+1) \) or \( j - 2\overline{m} \ge 2N_k+2 \) one has \( |\xi_j(\omega)| \le |j-2\overline{m}|/4 \), so \( |\theta_j(\omega) - 2\overline{m}| \ge 3|j-2\overline{m}|/4 \), and
\[ \sum_{\substack{ j - 2\overline{m} \le -(2N_k+1) \\ \text{or } j - 2\overline{m} \ge 2N_k+2 }} |s_k^\omega (\theta_j(\omega))|^2 \le \sum_{|\ell| \ge 2N_k+1} \frac{a^2}{(1 + (3|\ell|/4)^2)^2} \le \frac{\|S\|_2 ^2}{4k}, \]
where the last inequality follows from \eqref{eq:tail} since \(\{|\ell|\ge 2N_k+1\}\subset\{|\ell|\ge 2N_k\}\). Likewise, the membership \( T^{2\overline{m}}(\boldsymbol{\xi}(\omega))\in A_{\delta_k, N_k} \) means that for indices with \( j - 2\overline{m} \in \{-2N_k, \dots, 2N_k+1\} \) the point \( \theta_j(\omega)-2\overline{m} \) lies within distance \(\delta_k\) of an even integer, where \(S\) vanishes; hence
\[ \sum_{j - 2\overline{m} \in \{-2N_k, \dots, 2N_k+1\}} |s_k^\omega (\theta_j(\omega))|^2 \le (4N_k+2) \|S'\|_{\infty}^2 \delta_k^2 = \left(1 + \frac{1}{2 N_k}\right) \frac{\|S\|_2 ^2}{2 k}. \]
Combining the two estimates and using \( \|s_k^\omega\|_2=\|S\|_2 \),
\[ \sum_{j\in\mathbb Z} |s_k^\omega(\theta_j(\omega))|^2 \le \left( \frac{3}{4} + \frac{1}{4N_k} \right) \frac{\|S\|_2^2}{k} \le \frac{\|s_k^\omega\|_2^2}{k}, \]
the last bound holding for all \( N_k \ge 1 \). Taking \( \overline{\Omega} = \bigcap_{k \ge 1} \Omega_k \) yields the desired conclusion, since \( \mathbb{P}(\overline{\Omega})=1 \).
\end{proof}

\begin{remark}
The hypothesis on \( \xi_n \) in Theorem \ref{thm:ar_jitt_notstable} can be relaxed: the conclusion remains valid for any stationary mixing process for which, for every \(k\), there exist \(N_k\) and \(\delta_k(N_k)\) as defined in the previous proof such that
\[
\mathbb{P}(A_{\delta_k,N_k}) - \mathbb{P}(B_{N_k}^c) > 0,
\]
in the sense of \eqref{eq:pos_prob}.
\end{remark}
The previous proof shows that for each \( \omega \in \overline{\Omega} \) there exists a sequence of pathological bandlimited signals whose energy is concentrated precisely at the locations where the jitter \( \boldsymbol{\xi}(\omega) \) exhibits the bad pattern \( E_k \). This is an ad hoc deterministic construction, designed to disprove a worst-case property, and it crucially relies on aligning the signal with a specific shift \( 2\overline{m} \) of \( \boldsymbol{\xi} \). 

A natural question, left open here, is whether the AR(1)-jittered lattice stably samples a \emph{random} bandlimited signal with probability one - for instance a stationary bandlimited Gaussian process, for which the relevant comparison is between time-averaged signal power and the corresponding average sample power. The mechanism of Theorem \ref{thm:ar_jitt_notstable} does not apply in that setting, since a stationary signal cannot align with any
prescribed shift of the jitter path.


\section{Finite-dimensional analysis} \label{sec:fin_dim_analysis}

The negative result of Theorem \ref{thm:ar_jitt_notstable} concerns the full Paley-Wiener space and the infinite jittered lattice. In practice, however, one works with finitely many samples and a finite-dimensional approximation of the signal, and the question of stability must be reformulated accordingly. The goal of this section is to develop such a finite-dimensional counterpart and to identify regimes in which the associated sampling matrix is \emph{well-conditioned} with high probability. We remark that, unlike Theorem \ref{thm:ar_jitt_notstable}, the results of this section require only that \( \boldsymbol{\xi}_N \) be a centered Gaussian vector with \( \operatorname{var}(\xi_n) = \sigma_\xi^2 \); the autoregressive structure enters through the explicit form \eqref{eq:cholL} of the Cholesky factor \( \mathbf{L}_N \), but not through the resulting bounds.  

If \(s \in \mathrm{PW}_{[-1/2,1/2]} \), the Shannon-Whittaker-Nyquist-Kotelnikov theorem gives 
\[ s(t) = \sum_{n \in \mathbb{Z}}s(n)\sinc(t - n), \]
and in applications one typically resorts to the approximation 
\[ s(t) \approx \sum_{|n| \le N }s(n)\sinc(t - n) \]
for \(N\) sufficiently large. In this sense, \(s\) can be regarded as living in \( \mathrm{span}\{\sinc(\cdot - n)\}_{|n|\le N} \). The vector of jittered samples \( \widetilde{\mathbf{s}}_N \coloneqq (s(-N + \xi_{-N}),\dots ,s(N+\xi_{N}))^t \) is related to the non-jittered one \( \mathbf{s}_N \coloneqq (s(-N ),\dots ,s(N))^t \) through the random \( (2N+1) \times (2 N+1) \) sampling matrix
\begin{equation} \label{eq:sens_matr} \mathbf{A}^{(N)} \coloneqq \begin{pmatrix} \sinc(\xi_{-N}) & \sinc(-1 + \xi_{-N}) & \dots & \sinc(-2N + \xi_{-N}) \\ \sinc(1 + \xi_{-N+1}) & \sinc(\xi_{-N+1}) & \dots & \sinc(-2N+1 + \xi_{-N+1}) \\ \vdots & \vdots & \ddots & \vdots \\ \sinc(2N + \xi_N) & \sinc(2N - 1 + \xi_N) & \dots & \sinc( \xi_N)\end{pmatrix},  \end{equation} 
so that \( \widetilde{\mathbf{s}}_N = \mathbf{A}^{(N)} \mathbf{s}_N \). Then 
\[  \sigma_{\min}(\mathbf{A}^{(N)}) \|\mathbf{s}_N\|_{\ell^2} = \sigma_{\min}(\mathbf{A}^{(N)}) \|s\|_{L^2} \le \|\mathbf{A}^{(N)} \mathbf{s}_N\|_{\ell^2} = \|\widetilde{\mathbf{s}}_N\|_{\ell^2},    \]
where the first identity follows from Parseval. Rearranging yields \( \|s\|_{L^2} \le \sigma_{\min}(\mathbf{A}^{(N)})^{-1} \|\widetilde{\mathbf{s}}_N\|_{\ell^2} \) (cf.\ Definition \ref{eq:stable_sampling}), so the stability of finite-dimensional reconstruction reduces to the conditioning of \(\mathbf{A}^{(N)}\). Note that, in general, \( \mathbf{A}^{(N)} \) is neither symmetric nor Toeplitz. Observe also that \( \boldsymbol{\xi}_N \mapsto \text{det}(\mathbf{A}^{(N)} (\boldsymbol{\xi}_N))  \) is analytic and not identically zero, and thus \( \mathbb{P}(\{ \text{det}(\mathbf{A}^{(N)}) = 0 \})=0 \) for all \( N \), making the mere invertibility a trivial matter. We prove the following. 
\begin{proposition} \label{prop:prob_lb_sigmamin}
Let \( \mathbf{A}^{(N)} \) be the sensing matrix defined in \eqref{eq:sens_matr}, where
\( \boldsymbol{\xi}_N = (\xi_{-N}, \dots , \xi_N) \) is a centered Gaussian random vector with
\( \operatorname{var}(\xi_n) = \sigma_\xi^2 \le 10^{-2} \) for all \( |n| \le N \) and arbitrary correlation
structure. Let \( \eta \in (0,1) \), \( \epsilon \in (0, 2 - e^{16 \pi / 1000}) \) and set \( c \coloneqq \log(2-\epsilon)/\pi \).
Then \( \sigma_{\min} ( \mathbf{A}^{(N)}) \ge \epsilon \) with probability at least \( 1-\eta \) if
\begin{equation} \label{eq:cond_on_N} N \le \frac{1}{2} \left[ \frac{\log(1 - \eta)}{\log \left(1 - \sqrt{\frac{8}{\pi}} \, \dfrac{e^{-c^2/(2 \sigma_\xi ^2)}}{\frac{c}{\sigma_\xi} + \sqrt{\frac{c^2}{\sigma_\xi^2} + 2}} \right)} - 1   \right]. \end{equation}
\end{proposition}

\begin{proof}
Observe preliminarily that the assumptions on \(\epsilon\) and \( \sigma_\xi \) automatically imply that \begin{equation} \label{eq:komatsu_validity}
\sqrt{\frac{8}{\pi}} \, e^{-c^2/(2 \sigma_\xi ^2)} < \frac{c}{\sigma_\xi} + \sqrt{\frac{c^2}{\sigma_\xi^2} + 2}.
\end{equation}By Weyl's inequality for singular values, 
\begin{equation} \label{eq:sv_ineq} \begin{split} \sigma_{\min} (\mathbf{A}^{(N)}) & \ge \sigma_{\min} (\mathbf{I}_N) - \|\mathbf{A}^{(N)} - \mathbf{I}_N\|_{\text{op}} \\ & \ge 1 - \|\mathbf{A}^{(N)} - \mathbf{I}_N\|_{\text{op}}, \end{split} \end{equation} 
where \( \mathbf{I}_N \) denotes the \( (2N+1) \times (2N + 1) \) identity matrix. We now study \( \|\mathbf{A}^{(N)} - \mathbf{I}_N\|_{\text{op}} \). Introduce the auxiliary matrices 
\[ \mathbf{T}^{(N)} \coloneqq \begin{pmatrix} 0 & -1 & - 2 &  \dots & -2N \\ 1 & 0 & -1 & \dots & -2N+1 \\ \vdots & \vdots & \vdots & \ddots & \vdots  \\ 2N & 2N-1 & 2N-2 & \dots & 0  \end{pmatrix}, \quad \mathbf{D}_{\boldsymbol{\xi}_N ^r} \coloneqq \mathrm{diag}(\xi_{-N}^r, \dots , \xi_N ^r),  \] 
and observe that 
\begin{equation*} \mathbf{A}^{(N)} - \mathbf{I}_N = \sum_{r \ge 1} \frac{1}{r!}\mathbf{D}_{\boldsymbol{\xi}_N ^r} {\text{sinc} ^{(r)}} \circ \mathbf{T}^{(N)}, \end{equation*} 
where \(\sinc ^{(r)} \) is the \(r\)-th derivative of the sinc function and \(\circ\) denotes componentwise application of \(\sinc ^{(r)} \) to the entries of \( \mathbf{T}^{(N)} \). Setting \( \delta_N \coloneqq \max_{|n| \le N} |\xi_n| \), we obtain
\[\|\mathbf{A}^{(N)} - \mathbf{I}_N\|_{\text{op}} \le \sum_{r \ge 1} \frac{\delta_N ^r \|{\text{sinc} ^{(r)}} \circ \mathbf{T}^{(N)}\|_{\text{op}}}{r!} \le \sum_{r \ge 1} \frac{\delta_N ^r \pi^r}{r!} = e^{\delta_N \pi} - 1, \] 
where the last inequality uses the fact that the (infinite) Toeplitz matrix \( {\sinc ^{(r)}} \circ \mathbf{T}^{(\infty)} \) has Fourier symbol \( (2 \pi i \omega)^r \) on \( [-1/2,1/2] \), and that
\[ \|{\text{sinc} ^{(r)}} \circ \mathbf{T}^{(N)}\|_{\text{op}} \le \|{\text{sinc} ^{(r)}} \circ \mathbf{T}^{(\infty)}\|_{\text{op}} = \max_{\omega \in [-1/2,1/2] } |2 \pi \omega|^r = \pi ^r; \]
see Theorem 1.1 in \cite{boettcher_grudsky}. Substituting back into \eqref{eq:sv_ineq} gives \( \sigma_{\min} (\mathbf{A}^{(N)}) \ge 2 - e^{\delta_N \pi}  \). It remains to determine when \( 2 - e^{\delta_N \pi} \ge \epsilon > 0 \) with high probability, which is equivalent to requiring
\begin{equation} \label{eq:prob_lb} \mathbb{P}[ |\xi_{-N}| \le c, \ |\xi_{-N+1}| \le c, \ \dots \ , |\xi_N| \le c] \ge 1 - \eta \end{equation} 
for \(\eta\) small.

By the classical Šidák theorem\footnote{It is unclear to the authors whether the problem of finding a non-trivial non-asymptotic lower bound for the left-hand side of \eqref{eq:prob_lb} has been satisfactorily solved in the literature when the \( \xi_n \) are Gaussian and correlated. Without the absolute values, the ``excess'' probability caused by the correlations was estimated by Li and Shao in \cite{LiShao}, but they themselves remark that ``[...] it may be possible to have a comparisons of the two sided probabilities [with absolute values], [...] however, at this time, we are unable to find a right formulation [...]''.} \cite{sidak} we obtain
\[ \mathbb{P}[ |\xi_{-N}| \le c, \ |\xi_{-N+1}| \le c, \ \dots \ , |\xi_N| \le c] \ge \prod_{n=-N}^N \mathbb{P}[|\xi_n| \le c] = \left[ \mathrm{erf} \left( \frac{c}{\sigma_\xi \sqrt{2}} \right) \right]^{2N+1},   \] 
and Komatsu's inequality \cite{Dumbgen} finally yields
\begin{equation} \label{eq:after_komatsu} \left[ \mathrm{erf} \left( \frac{c}{\sigma_\xi \sqrt{2}} \right) \right]^{2N+1} \ge \left[1 - \sqrt{\frac{8}{\pi}} \, \frac{e^{-c^2/(2 \sigma_\xi ^2)}}{\frac{c}{\sigma_\xi} + \sqrt{\frac{c^2}{\sigma_\xi^2} + 2}}  \right]^{2N+1}.  \end{equation} 
For fixed \( c \) and \( \sigma_\xi \), the right-hand side of \eqref{eq:after_komatsu} tends to \(0\) as \(N \to \infty \), so the lower bound is asymptotically useless. In the finite-dimensional setting, however, it allows us to estimate how large the sensing matrix \( \mathbf{A}^{(N)} \) may be ``allowed to be'' in order to still have \( \sigma_{\min} (\mathbf{A}^{(N)}) \ge \epsilon \) with probability at least \(1 - \eta \). A direct calculation from \eqref{eq:after_komatsu} gives
\[ N \le \frac{1}{2} \left[ \frac{\log(1 - \eta)}{\log \left(1 - \sqrt{\frac{8}{\pi}} \, \dfrac{e^{-c^2/(2 \sigma_\xi ^2)}}{\frac{c}{\sigma_\xi} + \sqrt{\frac{c^2}{\sigma_\xi^2} + 2}}  \right)} - 1   \right]. \]
\end{proof}

\subsection{First-order Taylor expansion analysis} \label{sec:first_ord_analysis}
Proposition \ref{prop:prob_lb_sigmamin} is based on a uniform tail bound on \( \max_{|n| \le N} |\xi_n| \) and is consequently loose: it captures the worst single coordinate of \( \boldsymbol{\xi} \) rather than the collective effect on the operator norm. In the small-jitter regime \( \sigma_\xi \ll 1 \), typical of hardware applications, a sharper analysis is possible via first-order linearization. The sensing matrix \( \mathbf{A}^{(N)} \) may then be replaced with the first-order Taylor surrogate \begin{equation} \label{eq:ford_surr} \mathbf{B}^{(N)} \coloneqq \mathbf{I}_N + \mathbf{D}_{\boldsymbol{\xi}_N^1} \mathbf{S}_N, \end{equation} where \( (\mathbf{S}_N)_{mn} = (-1)^{m-n}/(m-n) \) for \( m \ne n \) and \( (\mathbf{S}_N)_{mm} = 0 \), obtained by evaluating the first derivative of \( \sinc \) at integer arguments. In signal-processing practice, \( \mathbf{B}^{(N)} \) often entirely replaces \( \mathbf{A}^{(N)} \).

\begin{proposition} \label{prop:linearized_prob_lb}
Let \( \mathbf{B}^{(N)} \) be the first-order Taylor surrogate defined in \eqref{eq:ford_surr},
where \( \boldsymbol{\xi}_N = (\xi_{-N}, \dots , \xi_N) \) is a centered
Gaussian random vector with \( \operatorname{var}(\xi_n) = \sigma_\xi^2 \) for all \( |n| \le N \)
and arbitrary correlation structure. Let
\( \epsilon, \eta \in (0,1) \). Then \( \sigma_{\min} ( \mathbf{B}^{(N)}) \ge \epsilon \) with
probability at least \( 1 - \eta \) if
\begin{equation} \label{eq:lin_uppbonN}
    N \le \frac{1}{4} \left[ \eta \exp\!\left((1-  \epsilon)^2 / (2 \sigma_\xi ^2 \pi^2)\right) - 2   \right],
\end{equation}
which should be compared to condition \eqref{eq:cond_on_N}.
\end{proposition}

\begin{proof}
The matrix \( \mathbf{S}_N \) is Toeplitz and skew-symmetric with symbol \( 2 \pi i \omega \) (cf.\ the proof of Proposition \ref{prop:prob_lb_sigmamin}), and therefore \(\|\mathbf{S}_N \|_{\text{op}} \le \pi \). What we would like, however, is a direct bound on \( \| \mathbf{D}_{\boldsymbol{\xi}_N^1} \mathbf{S}_N\|_{\text{op}}  \). We can write
\[ \mathbf{D}_{\boldsymbol{\xi}_N^1} \mathbf{S}_N = \sum_{ |n| \le N } g_n \, \mathrm{diag}((\mathbf{L} _N)_{:n}) \, \mathbf{S}_N \eqqcolon \sum_{ |n| \le N } g_n \widetilde{\mathbf{S}}_{N, n}, \]
where \((\mathbf{L} _N)_{:n}\) is the \(n\)-th column of the Cholesky decorrelation matrix \(\mathbf{L} _N\), \(\mathbf{L} _N \mathbf{L} _N^t = \boldsymbol{\Sigma}_N\), and \( (g_{-N}, \dots, g_N) \eqqcolon \mathbf{g}_N \sim \mathcal{N}(\mathbf{0}, \mathbf{I}_N) \). In the AR(1) case \eqref{eq:two-sidedAR1}, whose stationary variance is
\( \sigma_\xi^2 = \sigma_\epsilon^2/(1-\varphi^2) \), the matrix \( \mathbf{L} _N \) may be
obtained explicitly by ``unrolling'' the recurrence:
\begin{equation} \label{eq:cholL} \mathbf{L} _N = \begin{pmatrix}
\sigma_\xi & 0 & 0 & \cdots & 0 \\
\sigma_\xi \varphi & \sigma_\xi \sqrt{1 - \varphi^2} & 0 & \cdots & 0 \\
\sigma_\xi \varphi^2 & \sigma_\xi \varphi \sqrt{1 - \varphi^2}  & \sigma_\xi \sqrt{1 - \varphi^2} & \cdots & 0 \\
\vdots & \vdots & \vdots & \ddots & \vdots \\
\sigma_\xi \varphi^{2N} & \sigma_\xi \varphi^{2N-1} \sqrt{1 - \varphi^2} & \sigma_\xi \varphi^{2N-2} \sqrt{1 - \varphi^2} & \cdots & \sigma_\xi \sqrt{1 - \varphi^2}
\end{pmatrix}. \end{equation}
The argument below, however, only uses the identity
\( \sum_n (\mathbf{L}_N)_{mn}^2 = (\boldsymbol{\Sigma}_N)_{mm} = \sigma_\xi^2 \), which holds for
any Cholesky factor of a covariance matrix with constant diagonal.
To apply Theorem 4.1.1 in \cite{Tropp}, which addresses random series with matrix coefficients, we need to estimate
\[
v(\mathbf{D}_{\boldsymbol{\xi}_N^1} \mathbf{S}_N) \coloneqq \max \Bigl\{ \bigl\|\textstyle\sum_n \widetilde{\mathbf{S}}_{N, n} \widetilde{\mathbf{S}}_{N, n} ^t \bigr\|_{\text{op}}, \, \bigl\|\textstyle\sum_n \widetilde{\mathbf{S}}_{N, n} ^t \widetilde{\mathbf{S}}_{N, n}  \bigr\|_{\text{op}} \Bigr\}.
\]
For the first term, we have \( \sum_n \widetilde{\mathbf{S}}_{N, n} ^t \widetilde{\mathbf{S}}_{N, n} = -\sigma_\xi ^2 \mathbf{S}_N ^2 \), because \( \sum_{n} (\mathbf{L}_N)_{mn} ^2 = \sigma_\xi ^2 \) for all \(m\); hence \( \|\sum_n \widetilde{\mathbf{S}}_{N, n} ^t \widetilde{\mathbf{S}}_{N, n}\|_{\text{op}} \le \sigma_{\xi} ^2 \pi^2 \). For the second term, we observe that
\[ \begin{split} \sum_{|n| \le N} \widetilde{\mathbf{S}}_{N, n} \widetilde{\mathbf{S}}_{N, n} ^t & = - \sum_{|n| \le N} \mathrm{diag}((\mathbf{L}_N)_{:n}) \mathbf{S}_N^2 \mathrm{diag}((\mathbf{L}_N)_{:n}) \\ & = - \underbrace{\left[\sum_{|n| \le N} (\mathbf{L}_N)_{:n} (\mathbf{L}_N)_{:n} ^t \right]}_{\eqqcolon \widetilde{\mathbf{L}}_N} \odot \mathbf{S}_N ^2. \end{split}  \]
By construction \( (\widetilde{\mathbf{L}}_N)_{nn} = \sum_{j \le n} (\mathbf{L}_N)_{nj} ^2 = \sigma_\xi ^2 \), and the matrix \( \widetilde{\mathbf{L}}_N \) is positive semidefinite, being a sum of rank-one positive semidefinite matrices. Therefore
\[ \| \widetilde{\mathbf{L}}_N \odot \mathbf{S}_N ^2\|_{\text{op}} \le \max_n (\widetilde{\mathbf{L}}_N)_{nn} \, \|\mathbf{S}_N ^2 \|_{\text{op}} \le \sigma_\xi ^2 \pi^2, \] 
where the first inequality is Theorem 5.5.18 in \cite{horn_johnson}. We conclude that \( v(\mathbf{D}_{\boldsymbol{\xi}_N^1} \mathbf{S}_N) \le \sigma_\xi ^2 \pi^2 \). Formula (4.1.6) in Theorem 4.1.1 of \cite{Tropp}, together with the previous calculations, gives 
\begin{equation} \label{eq:linear_probbound} \mathbb{P} \left[  \| \mathbf{D}_{\boldsymbol{\xi}_N^1} \mathbf{S}_N\|_{\text{op}} \le 1-  \epsilon \right] \ge 1 - 2(2N+1) \exp\!\left(-(1-  \epsilon)^2 / (2 \sigma_\xi ^2 \pi ^2)\right), \end{equation} 
and since \( \sigma_{\min}(\mathbf{B}^{(N)}) \ge 1 - \|\mathbf{D}_{\boldsymbol{\xi}_N^1} \mathbf{S}_N\|_{\text{op}} \) by Weyl's inequality, the right-hand side of \eqref{eq:linear_probbound} is \( \ge 1 - \eta \) provided that 
\[
    N \le \frac{1}{4} \left[ \eta \exp\!\left((1-  \epsilon)^2 / (2 \sigma_\xi ^2 \pi^2)\right) - 2   \right].
\]
\end{proof}

\begin{remark}
The same techniques used in the proofs of Propositions \ref{prop:prob_lb_sigmamin} and \ref{prop:linearized_prob_lb} yield probabilistic upper bounds for \( \sigma_{\max}(\mathbf{A}^{(N)}) \) and \( \sigma_{\max}(\mathbf{B}^{(N)}) \), namely \( \sigma_{\max}(\mathbf{A}^{(N)})  \le e^{\delta_N \pi} \) and \( \sigma_{\max}(\mathbf{B}^{(N)}) \le 1 + \|\mathbf{D}_{\boldsymbol{\xi}_N^1} \mathbf{S}_N\|_{\text{op}} \). Consequently, on the events considered in Propositions \ref{prop:prob_lb_sigmamin} and \ref{prop:linearized_prob_lb} respectively,
\[ 1 \le \kappa(\mathbf{A}^{(N)}) = \frac{\sigma_{\max}(\mathbf{A}^{(N)})}{\sigma_{\min}(\mathbf{A}^{(N)})} \le \frac{e^{\delta_N \pi}}{2- e^{\delta_N \pi}}, \qquad 1 \le \kappa(\mathbf{B}^{(N)}) \le \frac{2 - \epsilon}{\epsilon}, \] 
where the lower bounds are immediate from \( \sigma_{\min} \le \sigma_{\max} \). These provide information not only on the invertibility of \( \mathbf{A}^{(N)} \) and \( \mathbf{B}^{(N)} \), but also on their condition numbers, and hence on their numerical stability. In particular, the regions depicted in Figure \ref{fig:sidak_prob} are also regions of high-probability control on \( \kappa(\mathbf{A}^{(N)}) \) and \( \kappa(\mathbf{B}^{(N)}) \).
\end{remark}

\begin{figure}[h]
\begin{subfigure}{.5\textwidth}
  \centering
  \includegraphics[width=1\linewidth]{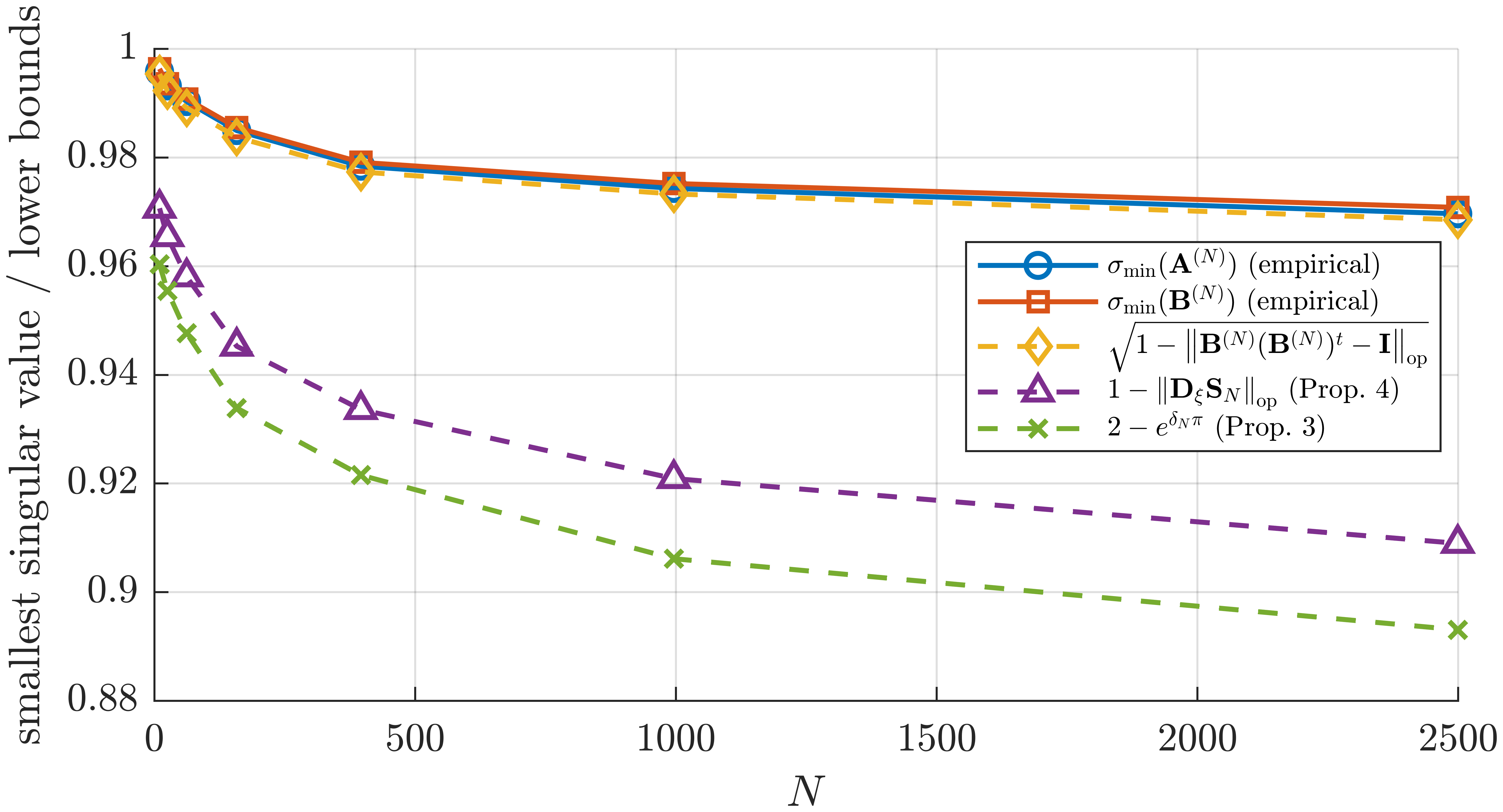}
  \caption{\( \sigma_\xi = 0.01 \), \( \varphi = 0.99 \)}
\end{subfigure}%
\begin{subfigure}{.5\textwidth}
  \centering
  \includegraphics[width=1\linewidth]{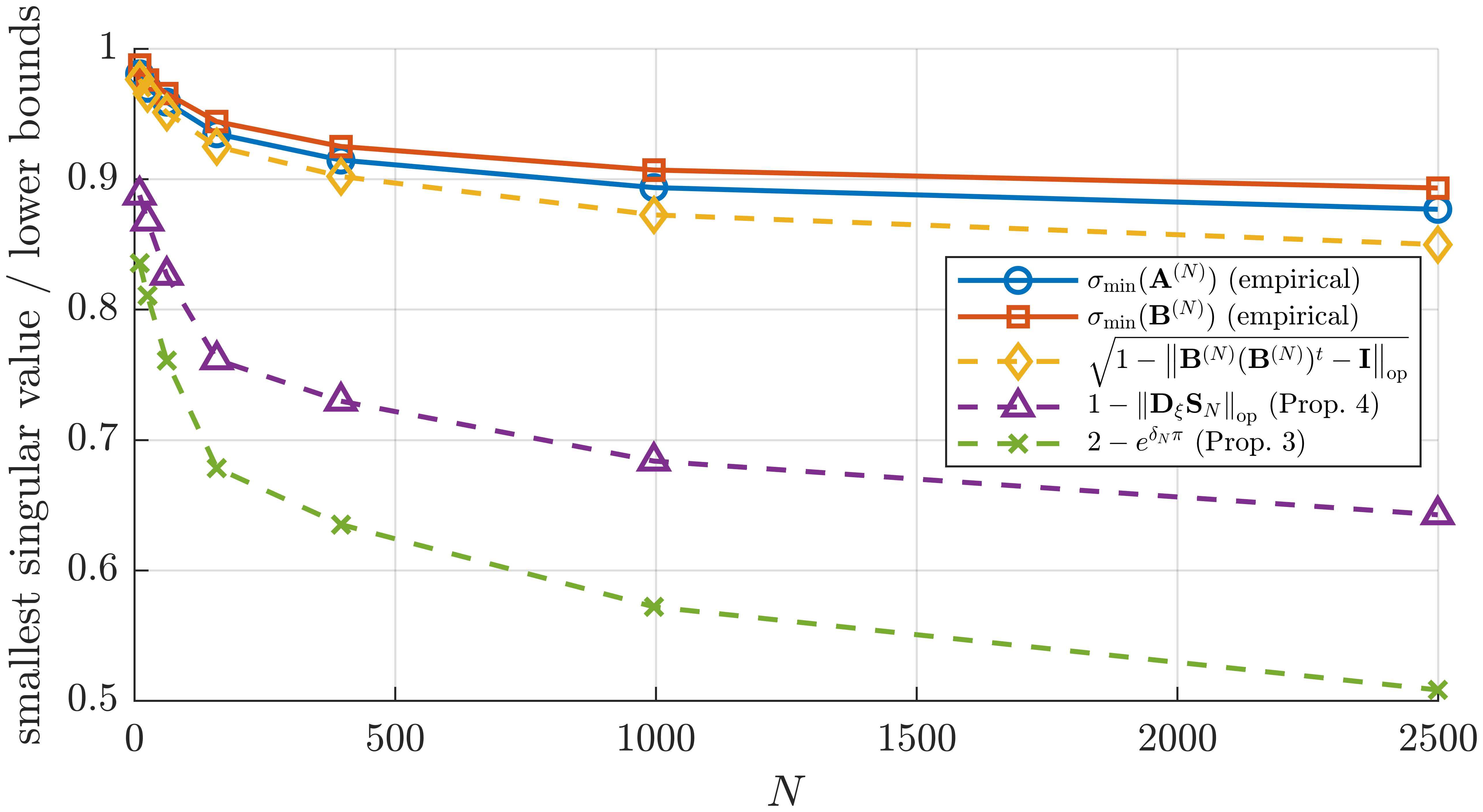}
  \caption{\( \sigma_\xi = 0.04 \), \( \varphi = 0.99 \)}
\end{subfigure}
\caption{The figures above show that the obtained bounds are loose.}
\label{fig:sigma_min_est}
\end{figure}

The figures \ref{fig:sigma_min_est} show that the bounds \( \sigma_{\min} (\mathbf{A}^{(N)}) \ge 2 - e^{\delta_N \pi} \) and \( \sigma_{\min} (\mathbf{B}^{(N)}) \ge 1 - \| \mathbf{D}_{\boldsymbol{\xi}^1 _N} \mathbf{S}_N \|_\text{op}  \) are rather loose. The same figures numerically suggest that a sharper lower bound for \( \sigma_{\min} (\mathbf{B}^{(N)}) \) is \( \sqrt{1 - \| \mathbf{I}_N  - \mathbf{B}^{(N)} (\mathbf{B}^{(N)})^t\|_\text{op} } \).

An analytic justification proceeds by exploiting the skew-symmetry \( \mathbf{S}_N^t = - \mathbf{S}_N \), which gives
\[
\mathbf{B}^{(N)} (\mathbf{B}^{(N)})^t - \mathbf{I}_N = [\mathbf{D}_{\boldsymbol{\xi}^1_N}, \mathbf{S}_N] - \mathbf{D}_{\boldsymbol{\xi}^1_N} \mathbf{S}_N^2 \mathbf{D}_{\boldsymbol{\xi}^1_N},
\]
where \( [\cdot,\cdot] \) denotes the commutator. The commutator is again a Gaussian series in the variables \( \{g_n\} \) of Section \ref{sec:first_ord_analysis}, with coefficient matrices \( \mathbf{C}_n \coloneqq \mathbf{S}_N \, \mathrm{diag}((\mathbf{L}_N)_{:n}) - \mathrm{diag}((\mathbf{L}_N)_{:n})\, \mathbf{S}_N \), and is therefore amenable to the same Theorem 4.1.1 of \cite{Tropp} that we used in the proof of Proposition \ref{prop:linearized_prob_lb}. The entries of \( \mathbf{C}_n \) involve differences of entries of \( \mathbf{L}_N \), which encode the AR(1) coherence between adjacent jitter samples and are small when \( \varphi \approx 1 \); these cancellations heuristically explain the empirically observed sharpness of the Gram-matrix bound in the highly correlated regime relevant to the ADC application. Translating them into an operator-norm bound with an explicit \( \varphi \)-dependent improvement over Proposition \ref{prop:linearized_prob_lb}, however, appears to require non-standard Schur-multiplier estimates that lie outside the scope of this note.

The figures below show a qualitative illustration of \eqref{eq:cond_on_N} and \eqref{eq:lin_uppbonN} for different values of \( \sigma_\xi\) and \( \epsilon \). 
\begin{figure}[h]
\begin{subfigure}{.26\textwidth}
  \centering
  \includegraphics[width=1\linewidth]{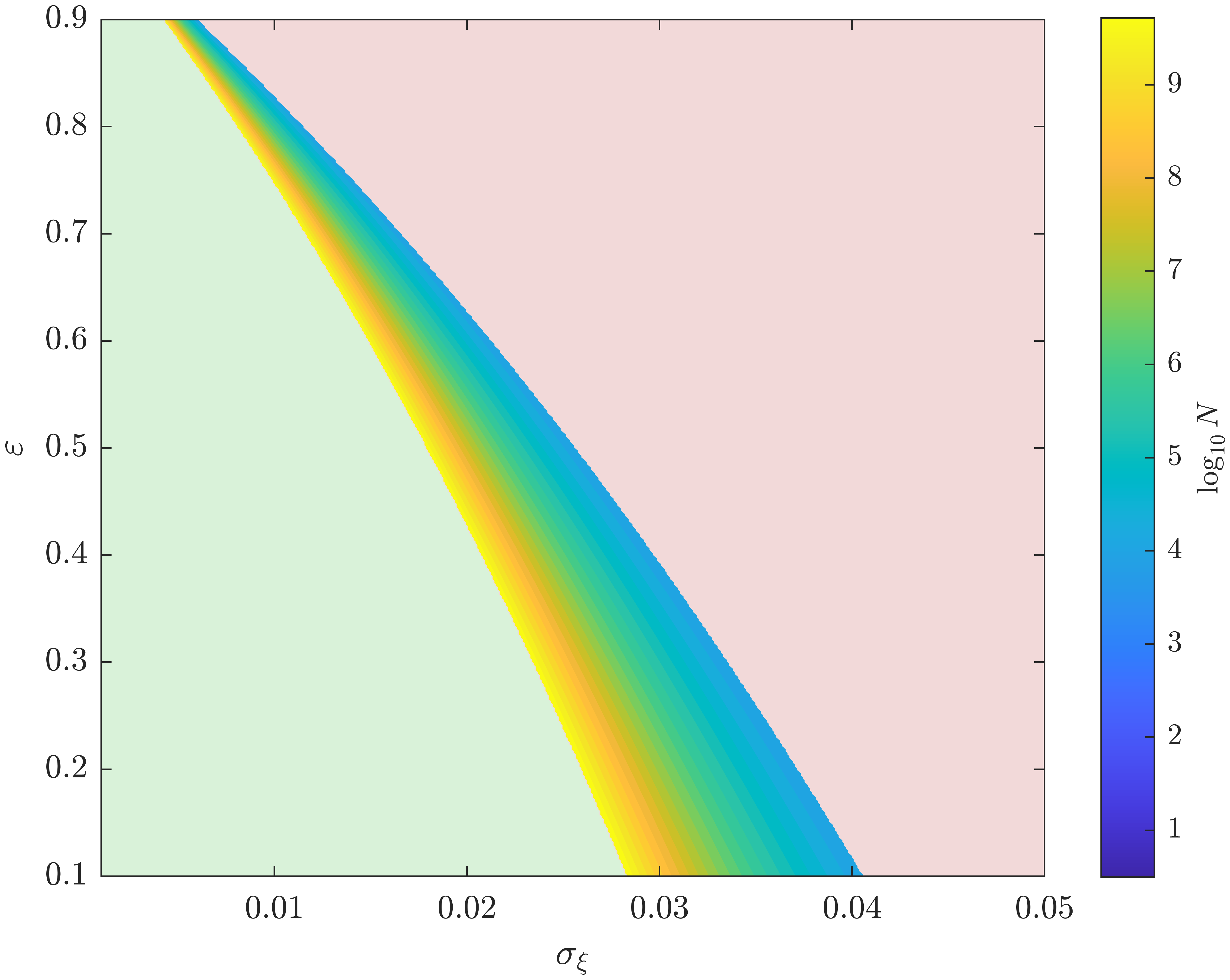}
\end{subfigure}%
\begin{subfigure}{.26\textwidth}
  \centering
  \includegraphics[width=1\linewidth]{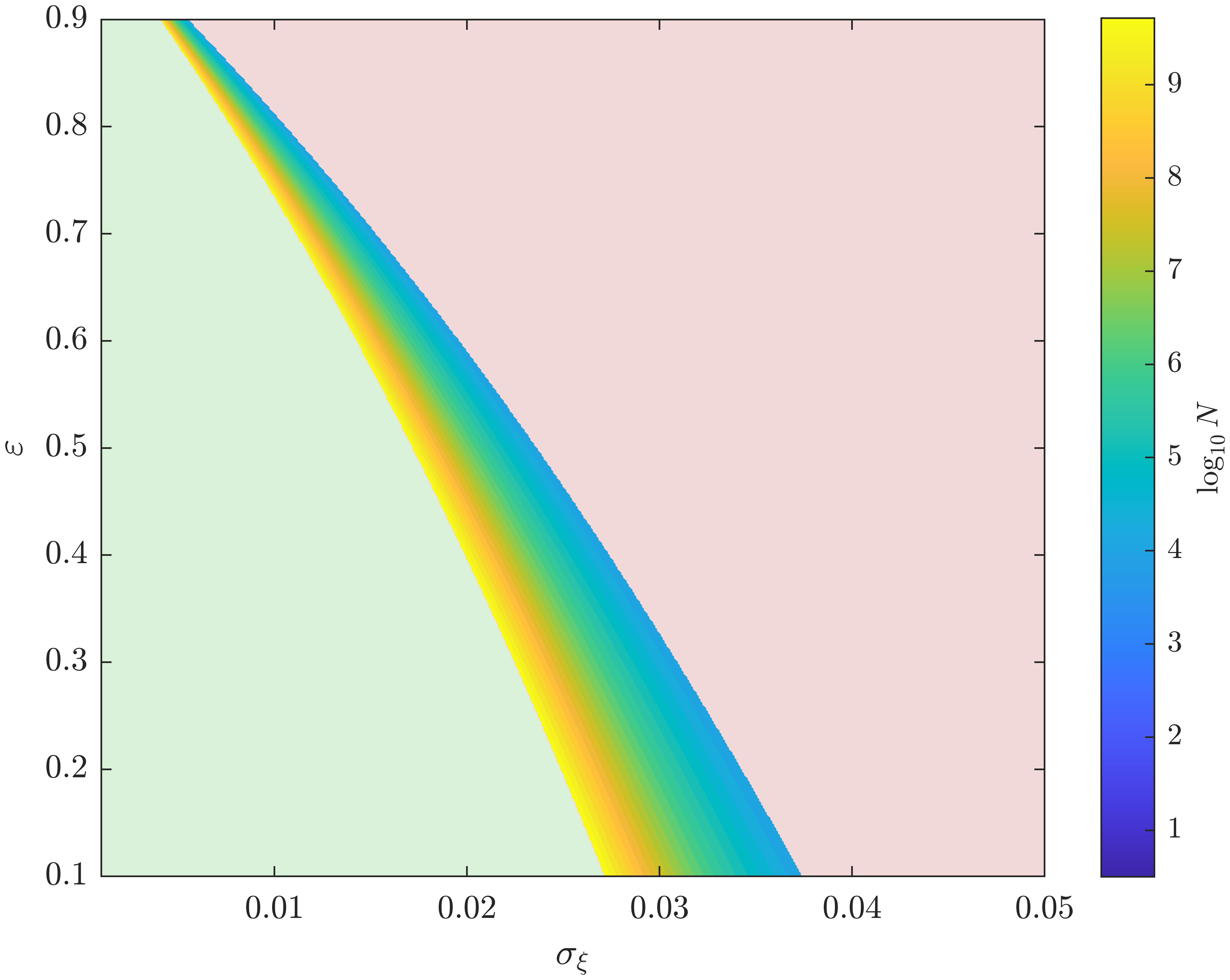}
\end{subfigure}
\begin{subfigure}{.26\textwidth}
  \centering
  \includegraphics[width=1\linewidth]{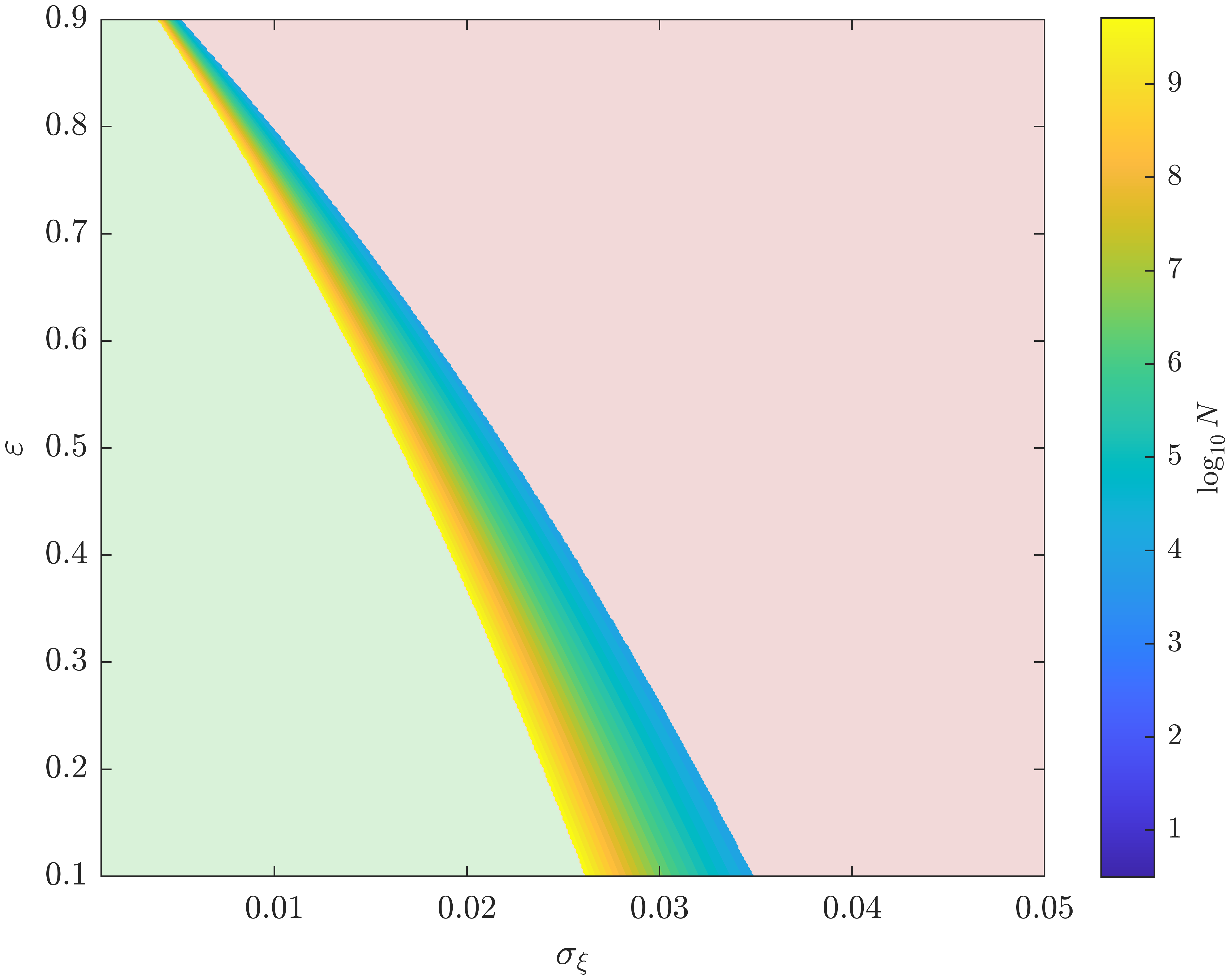}
\end{subfigure}%
\begin{subfigure}{.26\textwidth}
  \centering
  \includegraphics[width=1\linewidth]{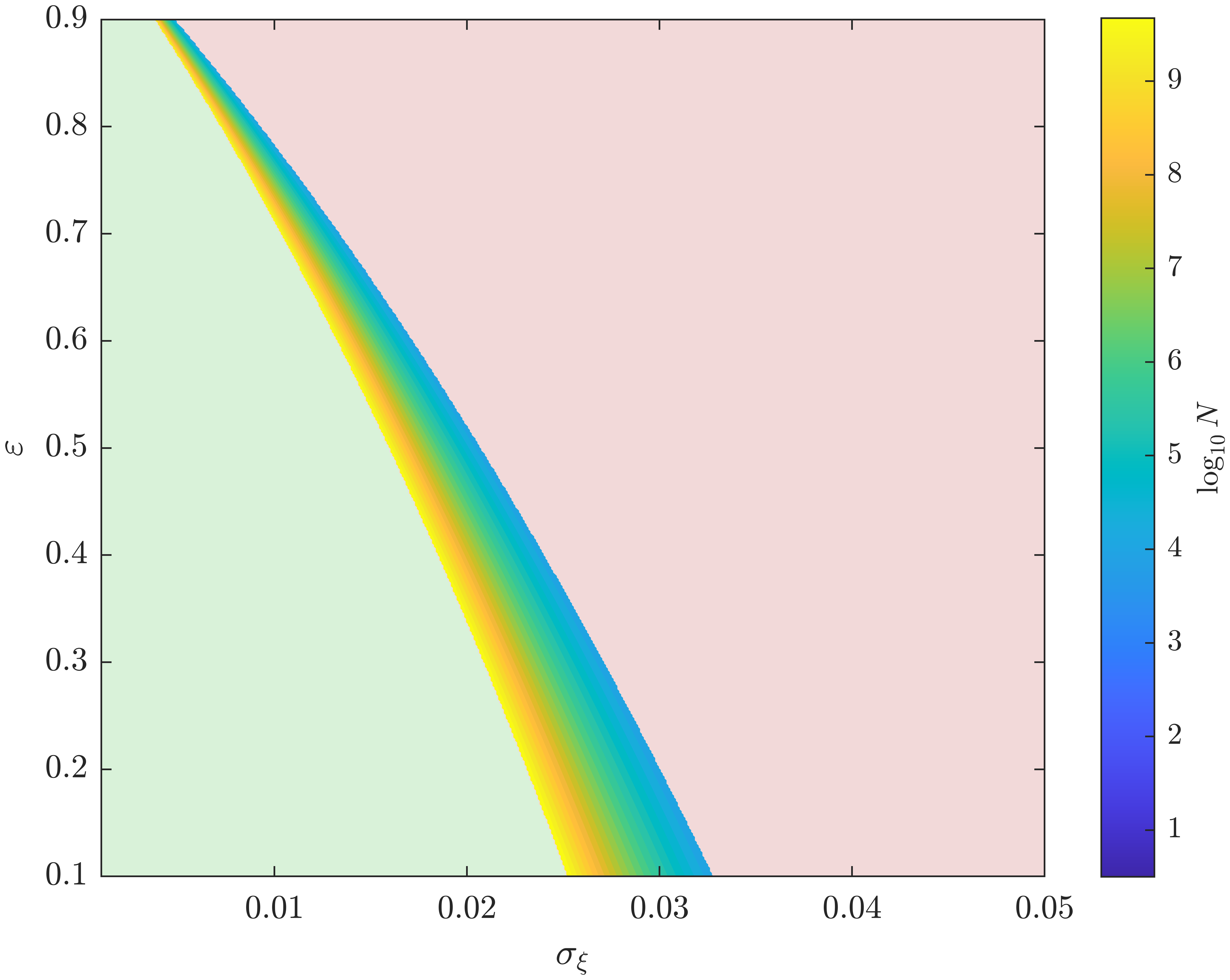}
\end{subfigure}
\begin{subfigure}{.26\textwidth}
  \centering
  \includegraphics[width=1\linewidth]{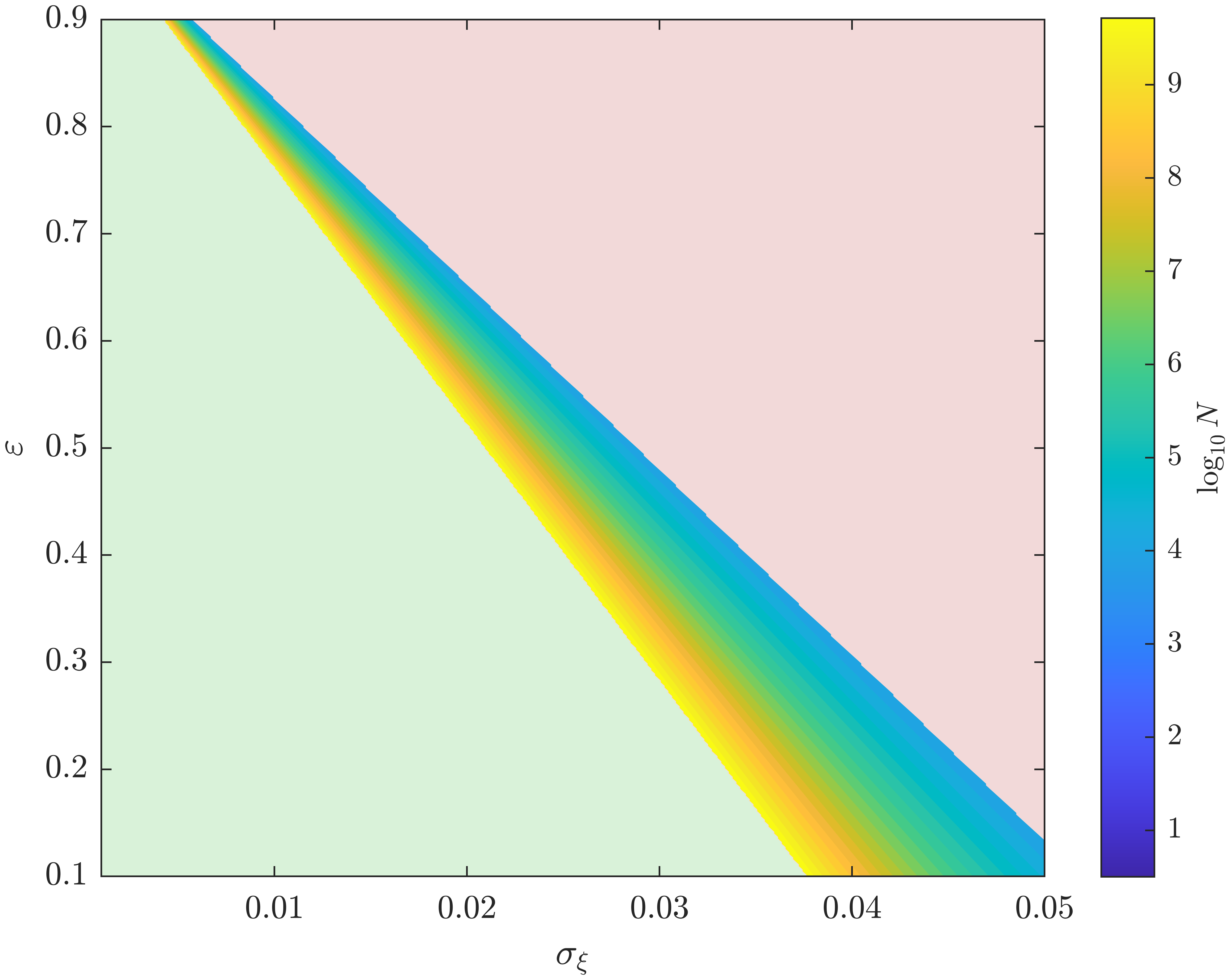}
  \caption{\(\eta = 10^{-2} \)}
  \label{fig:sfig1}
\end{subfigure}%
\begin{subfigure}{.26\textwidth}
  \centering
  \includegraphics[width=1\linewidth]{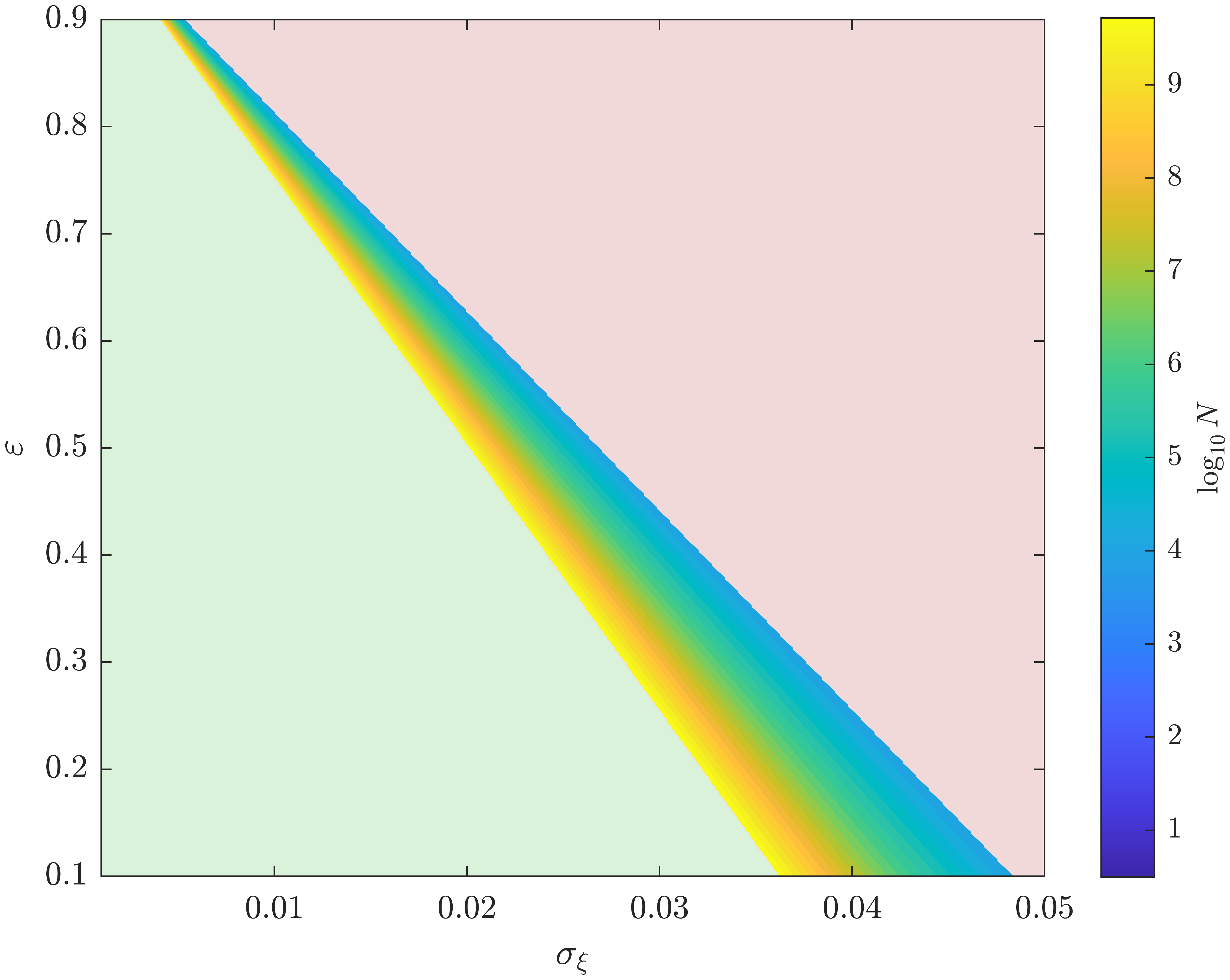}
  \caption{\(\eta = 10^{-3} \)}
  \label{fig:sfig2}
\end{subfigure}
\begin{subfigure}{.26\textwidth}
  \centering
  \includegraphics[width=1\linewidth]{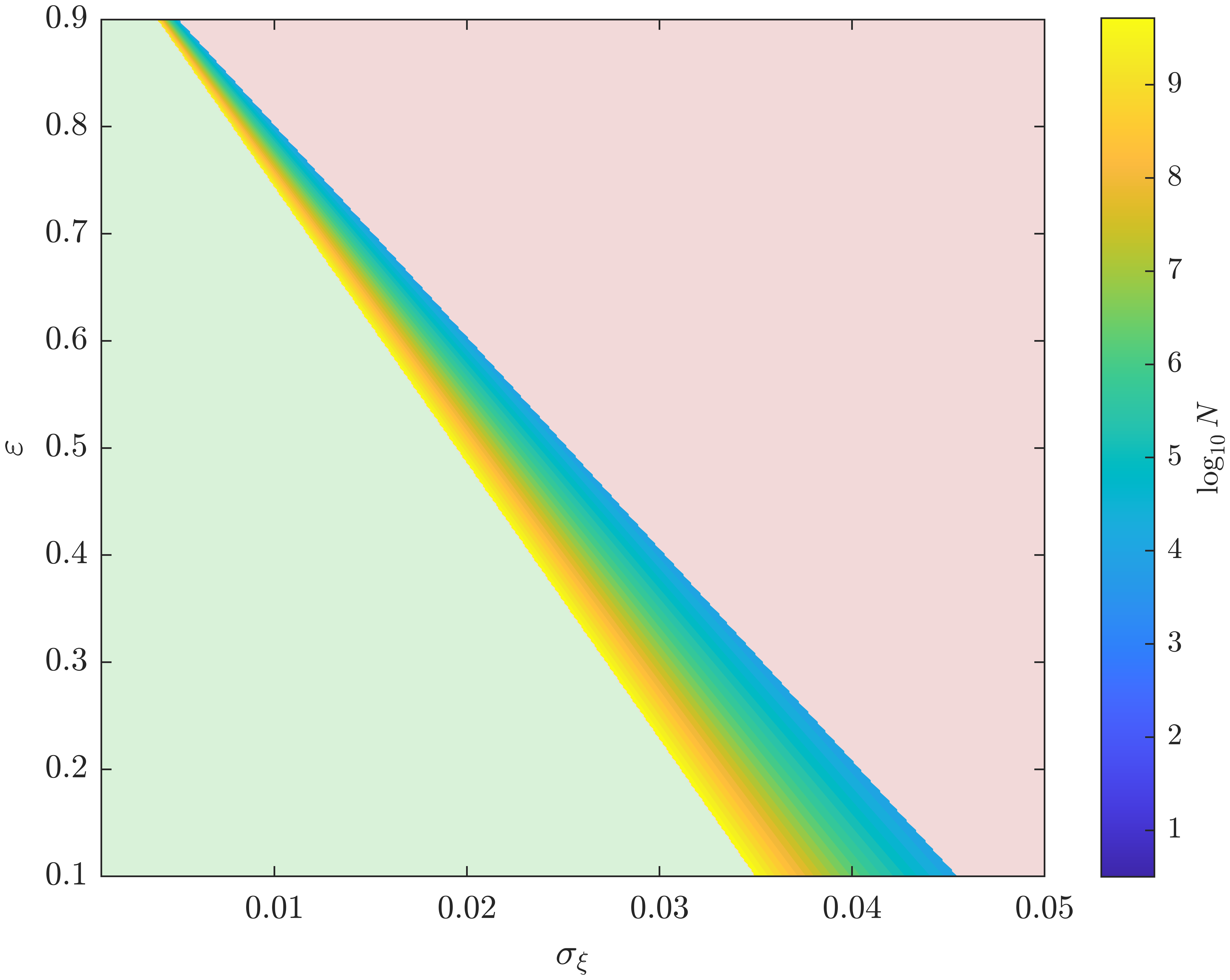}
  \caption{\(\eta = 10^{-4} \)}
  \label{fig:sfig1}
\end{subfigure}%
\begin{subfigure}{.26\textwidth}
  \centering
  \includegraphics[width=1\linewidth]{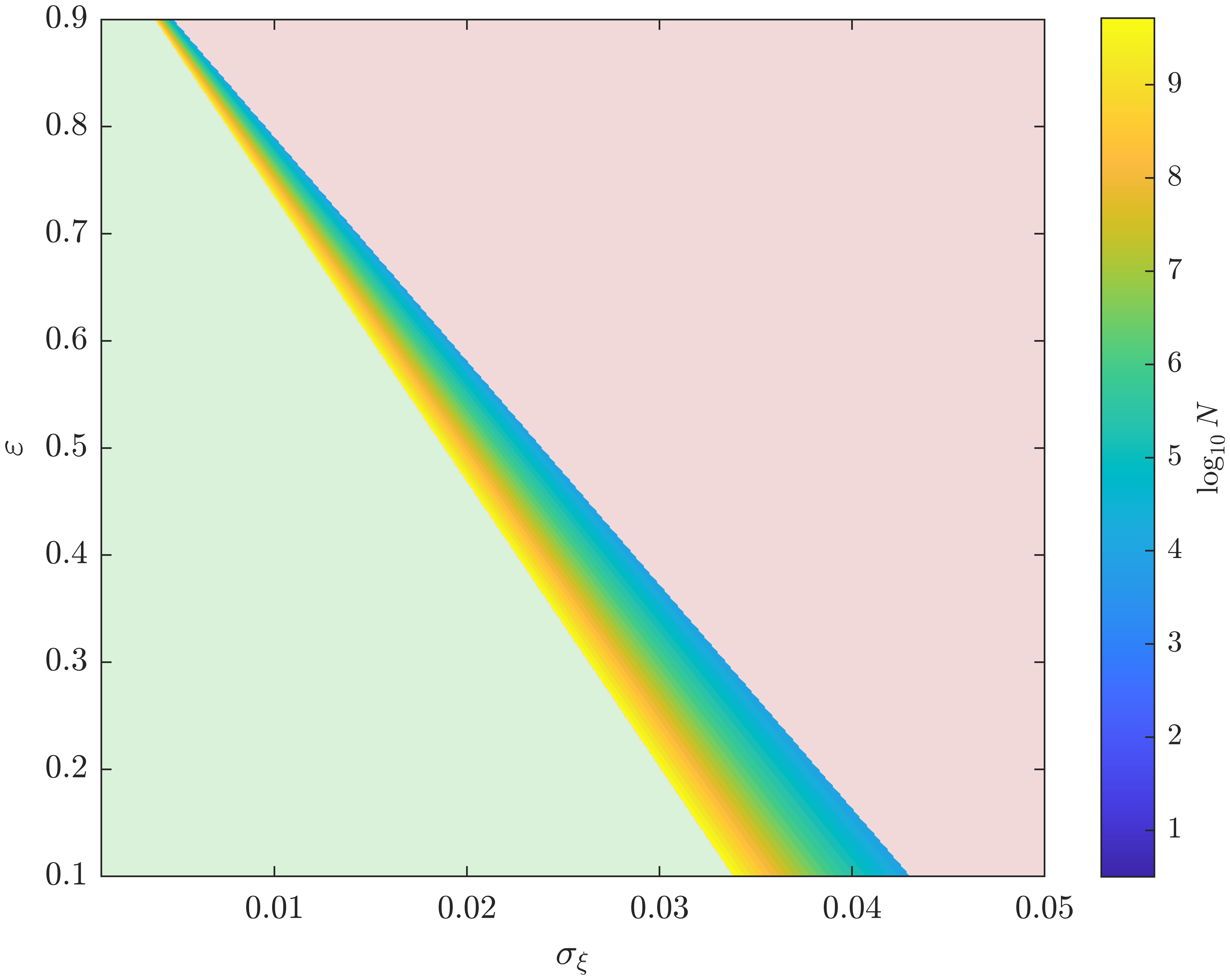}
  \caption{\(\eta = 10^{-5} \)}
  \label{fig:sfig2}
\end{subfigure}
\caption{Contour plots of the right-hand side of \eqref{eq:cond_on_N} (first row) and \eqref{eq:lin_uppbonN} (second row) as functions of \(\sigma_\xi\) and \(\epsilon\). In the pink region the bound falls below \( 10^3 \), while in the light green region it exceeds \( 10^{10} \). We remark that, in ADC applications, \( \sigma_{\xi} \) represents the average deviation from the ideal sampling instants \(nT_s\), so that \( \sigma_\xi  \ll T_s \) and \( \sigma_\xi / T_s \) is typically of the order of \(1 \%\) or less. Higher values of \( \sigma_\xi \) may correspond to harsher scenarios, but more likely reflect poor hardware designs.}
\label{fig:sidak_prob}
\end{figure}

\begin{remark} Real-world systems do not allow for infinite numerical precision, and typically some form of quantization follows the sampling operation in ADCs. If \(Q\) is an infinite uniform quantizer and \( Q(\mathbf{M}) \) is its component-wise application to all the coefficients of a matrix \( \mathbf{M} \), it is straightforward to extend both Propositions \ref{prop:prob_lb_sigmamin} and \ref{prop:linearized_prob_lb} to \( Q(\mathbf{M}) \) by means of a (deterministic) bound on \( \| Q(\mathbf{M}) - \mathbf{M} \|_{\text{op}} \) that follows from finite-dimensional norm equivalence and quantizer resolution. If the quantizer is finite within a finite range \( [-R,R] \), the bound on \( \| Q(\mathbf{M}) - \mathbf{M} \|_{\text{op}} \) will be ``probabilistic'', i.e. for example subjected to the probability of \( |\xi_i| \le R \) for all \(i=-N, \dots ,N\). \end{remark}

\section{Numerical Experiments} \label{sec:num_exp}

In this section we complement the analytical results of Section \ref{sec:fin_dim_analysis} with Monte Carlo simulations. Our aims are threefold: to describe the empirical distributions of \( \sigma_{\min} \) and of the condition number for both \( \mathbf{A}^{(N)} \) and \( \mathbf{B}^{(N)} \), to assess how these distributions depend on the autoregressive parameter \( \varphi \), and to quantify the conservativeness of the bounds obtained in Propositions \ref{prop:prob_lb_sigmamin} and \ref{prop:linearized_prob_lb}. The first two aims are addressed in Section \ref{subsec:svd_condn} while the third one in Section \ref{subsec:bound_qual}.

\subsection{Singular values and condition number distributions}\label{subsec:svd_condn}
In Figure \ref{fig:ecdf} we display the empirical cumulative distribution function (ECDF) of \( \sigma_{\min} (\mathbf{A}^{(N)}) \), \( \sigma_{\min} (\mathbf{B}^{(N)}) \), \( \kappa (\mathbf{A}^{(N)}) \) and \( \kappa (\mathbf{B}^{(N)}) \) respectively with \(N=1000\), derived by simulating \(1000\) matrix instances for five different correlation levels \( \varphi \in \{0.1, 0.25, 0.5, 0.9, 0.99\} \) and \( \sigma_\xi = 0.04\). The plots show that both random matrices \( \mathbf{A}^{(1000)} \) and \( \mathbf{B}^{(1000)} \) are very well-conditioned with high probability in all the correlation regimes tested (including the high-correlation regime, \( \varphi \approx 1\), which is of particular interest to ADC applications), consistently with Propositions \ref{prop:prob_lb_sigmamin} and \ref{prop:linearized_prob_lb}. The distributions relative to \( \mathbf{A}^{(1000)} \) and \( \mathbf{B}^{(1000)} \) are moreover very similar, which supports the use of the first-order surrogate \eqref{eq:ford_surr} in place of the full sensing matrix. Finally, the higher the correlation, the larger \( \sigma_{\min} (\mathbf{A}^{(N)}) \) and \( \sigma_{\min} (\mathbf{B}^{(N)}) \), with little-to-no overlapping between the ECDF curves; we remark that this dependence on \( \varphi \) is not captured by the bounds of Section \ref{sec:fin_dim_analysis}, which only involve the marginal variance \( \sigma_\xi^2 \). 

\begin{figure}[h]
  \centering
  \includegraphics[width=1\linewidth]{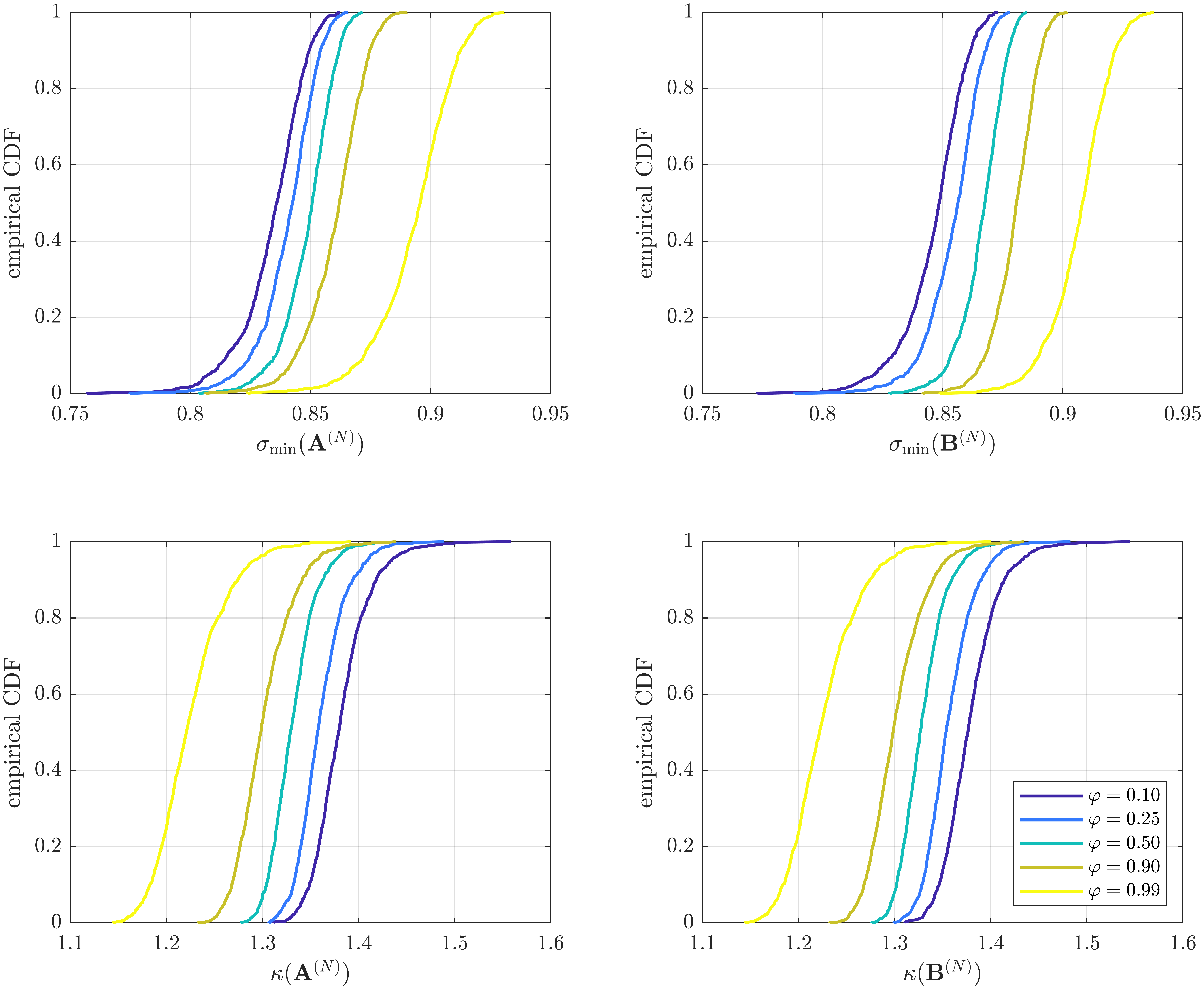}
\caption{Empirical cumulative distribution functions obtained via Monte Carlo.}
\label{fig:ecdf}
\end{figure}

\subsection{Lower bounds quality}\label{subsec:bound_qual}
As Figure \ref{fig:sigma_min_est} already suggests, the lower bounds used in Propositions \ref{prop:prob_lb_sigmamin} and \ref{prop:linearized_prob_lb} are somewhat crude, at least for ``large'' values of the process standard deviation \( \sigma_\xi \), and analytically deriving \emph{how} loose they are is beyond the scope of this work. Such bounds are tractable proxies that made it possible to rigorously demonstrate the well-conditioning of the jittered matrix with high probability. This self-contained subsection is meant to numerically assess the gap between such proxies and the true underlying quantities, as a function of \( \varphi \) and the matrix size \( N \).

Figure \ref{fig:looseness_study} displays the results of this assessment. We ran two separate Monte Carlo simulations in which we fixed the jitter process standard deviation to \( \sigma_\xi = 0.01 \) (top) and \( \sigma_\xi = 0.04 \) (bottom), and swept over the matrix dimensions and the geometric factor \( \varphi \). Notice that for an AR(1) process \( \sigma_\epsilon^2 = (1 - \varphi^2)\sigma_\xi^2 \), so that when \( \varphi \to 1 \) the innovation noise variance shrinks to \( 0 \); in this case the process degenerates into the uncorrelated case \( \xi_i \approx \xi_{-N} \) for all \( i \). In the heatmaps we plotted the mean relative gap (MRG)
\[ \frac{1}{N_{\text{MC iters}}} \sum_{n=1}^{N_{\text{MC iters}}} \frac{\sigma_{\min} - s}{\sigma_{\min}} \]
with \( N_{\text{MC iters}} = 250 \), where \( s = 2 - e^{\delta_N \pi} \) is compared against \( \sigma_{\min} = \sigma_{\min}(\mathbf{A}^{(N)}) \), while \( s = 1 - \|\mathbf{D}_\xi \mathbf{S}_N\|_{\text{op}} \) and \( s = \sqrt{1 - \|\mathbf{I}_N - \mathbf{B}^{(N)}(\mathbf{B}^{(N)})^t\|_{\text{op}}} \) are compared against \( \sigma_{\min} = \sigma_{\min}(\mathbf{B}^{(N)}) \).

\begin{figure}[h]
  \centering
  \includegraphics[width=1\linewidth]{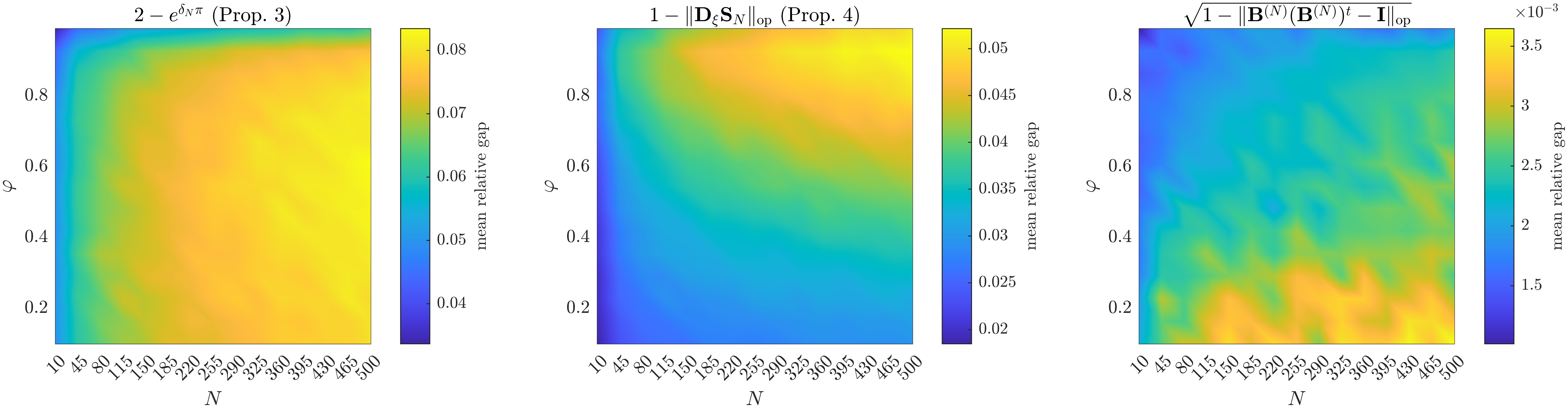}
  \includegraphics[width=1\linewidth]{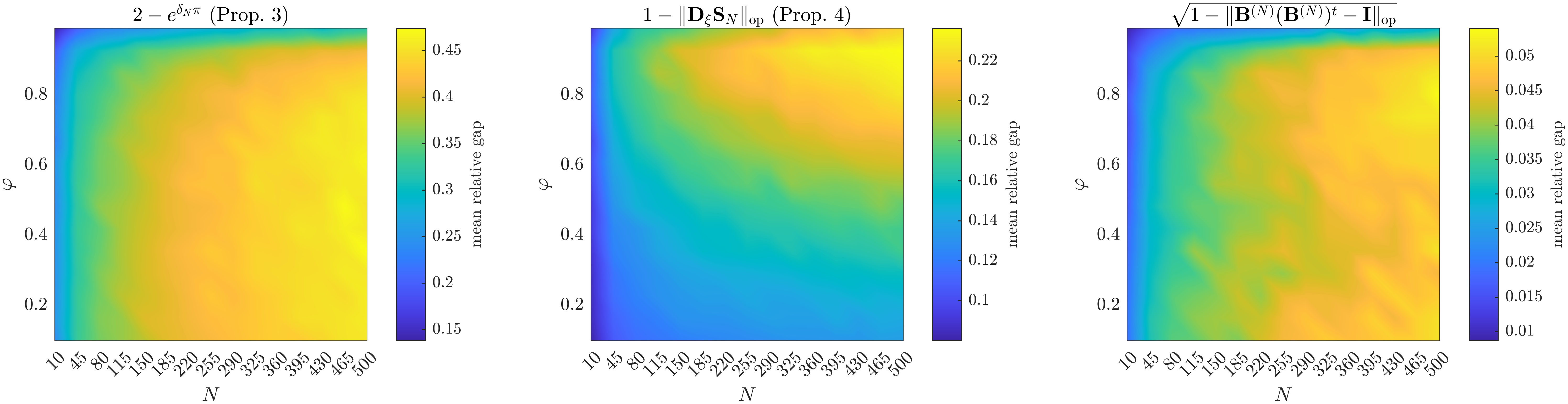}
\caption{MRG of the different lower bounds obtained in Section \ref{sec:fin_dim_analysis}.}
\label{fig:looseness_study}
\end{figure}

Figure \ref{fig:looseness_study} confirms the preliminary conclusions drew from Figure \ref{fig:sigma_min_est}, including the degradation of the lower bounds performance as the process standard deviation increases.

\subsection{Jitter cumulative correlation length and relation to the upper bounds in Propositions \ref{prop:prob_lb_sigmamin} and \ref{prop:linearized_prob_lb}}\label{subsec:corr_len}

The cumulative correlation length \(\tau\) of a sequence indexed by \( \mathbb{Z} \) is defined as the bilateral sum of its correlation coefficients, and it measures the degree of self-similarity of such a sequence; or, in other words, it provides an interval length within which all samples may be approximately considered as ``highly correlated''. For the process \eqref{eq:two-sidedAR1} it is easy to see that \( \tau(\varphi)=(1+\varphi)/(1-\varphi) \). In this section we test whether the results from Section \ref{sec:fin_dim_analysis} remain informative when compared against the correlation length, i.e. whether \(N\) can be large enough to span one or many correlation lengths. The quantity of interest will be the ratio of the upper bounds in Propositions \ref{prop:prob_lb_sigmamin} and \ref{prop:linearized_prob_lb} to \( \tau \). Notice that in our case \(N/\tau\) is a strictly decreasing function of \( \varphi \) over \( [0,1) \). Tables \ref{tab:ratio_prop3} and \ref{tab:ratio_prop4} show that the ratio \(N/\tau\) exceeds \(100\) throughout the tested range of \( \varphi \) in many setups of interest, and by several orders of magnitude in most of them. 

\begin{table}[h]
\centering
\caption[Ratio thresholds, Proposition 3]{(Rounded) Values of \( \varphi \) at which \( N/\tau = 1 \) and \( N/\tau = 100 \) (reported as \( \cdot / \, \cdot \)), for the bound of Proposition \ref{prop:prob_lb_sigmamin}. An entry at the right endpoint of the tested range indicates that the ratio exceeds the corresponding threshold throughout \( \varphi \in [0,1) \) (and analogously with entries at the left endpoint).}
\label{tab:ratio_prop3}
\begin{tabular}{cc|cccc}
\toprule
\( \epsilon \) & \( \eta \) & \( \sigma_\xi = 0.005 \) & \( \sigma_\xi = 0.01 \) & \( \sigma_\xi = 0.025 \) & \( \sigma_\xi = 0.04 \) \\
\midrule
\multirow{4}{*}{\( 0.05 \)}
 & \( 10^{-2} \) & \( 0.9999 \, / \, 0.9999 \) & \( 0.9999 \, / \, 0.9999 \) & \( 0.9999 \, / \, 0.9999 \) & \( 0.9999 \, / \, 0.9957 \) \\
 & \( 10^{-3} \) & \( 0.9999 \, / \, 0.9999 \) & \( 0.9999 \, / \, 0.9999 \) & \( 0.9999 \, / \, 0.9999 \) & \( 0.9995 \, / \, 0.9575 \) \\
 & \( 10^{-4} \) & \( 0.9999 \, / \, 0.9999 \) & \( 0.9999 \, / \, 0.9999 \) & \( 0.9999 \, / \, 0.9999 \) & \( 0.9957 \, / \, 0.6426 \) \\
 & \( 10^{-5} \) & \( 0.9999 \, / \, 0.9999 \) & \( 0.9999 \, / \, 0.9999 \) & \( 0.9999 \, / \, 0.9999 \) & \( 0.9570 \, / \, 0 \) \\
\midrule
\multirow{4}{*}{\( 0.25 \)}
 & \( 10^{-2} \) & \( 0.9999 \, / \, 0.9999 \) & \( 0.9999 \, / \, 0.9999 \) & \( 0.9999 \, / \, 0.9999 \) & \( 0.9965 \, / \, 0.7064 \) \\
 & \( 10^{-3} \) & \( 0.9999 \, / \, 0.9999 \) & \( 0.9999 \, / \, 0.9999 \) & \( 0.9999 \, / \, 0.9999 \) & \( 0.9658 \, / \, 0 \) \\
 & \( 10^{-4} \) & \( 0.9999 \, / \, 0.9999 \) & \( 0.9999 \, / \, 0.9999 \) & \( 0.9999 \, / \, 0.9999 \) & \( 0.6821 \, / \, 0 \) \\
 & \( 10^{-5} \) & \( 0.9999 \, / \, 0.9999 \) & \( 0.9999 \, / \, 0.9999 \) & \( 0.9999 \, / \, 0.9999 \) & \( 0 \, / \, 0 \) \\
\midrule
\multirow{4}{*}{\( 0.75 \)}
 & \( 10^{-2} \) & \( 0.9999 \, / \, 0.9999 \) & \( 0.9999 \, / \, 0.9999 \) & \( 0 \, / \, 0 \) & \( N \le 0 \) \\
 & \( 10^{-3} \) & \( 0.9999 \, / \, 0.9999 \) & \( 0.9999 \, / \, 0.9999 \) & \( N \le 0 \) & \( N \le 0 \) \\
 & \( 10^{-4} \) & \( 0.9999 \, / \, 0.9999 \) & \( 0.9999 \, / \, 0.9999 \) & \( N \le 0 \) & \( N \le 0 \) \\
 & \( 10^{-5} \) & \( 0.9999 \, / \, 0.9999 \) & \( 0.9999 \, / \, 0.9999 \) & \( N \le 0 \) & \( N \le 0 \) \\
\bottomrule
\end{tabular}
\end{table}

\begin{table}[h]
\centering
\caption[Ratio thresholds, Proposition 3]{Same as in table \ref{tab:ratio_prop3}, but for Proposition \ref{prop:linearized_prob_lb}.}
\label{tab:ratio_prop4}
\begin{tabular}{cc|cccc}
\toprule
\( \epsilon \) & \( \eta \) & \( \sigma_\xi = 0.005 \) & \( \sigma_\xi = 0.01 \) & \( \sigma_\xi = 0.025 \) & \( \sigma_\xi = 0.04 \) \\
\midrule
\multirow{4}{*}{\( 0.05 \)}
 & \( 10^{-2} \) & \( 0.9999 \, / \, 0.9999 \) & \( 0.9999 \, / \, 0.9999 \) & \( 0.9999 \, / \, 0.9999 \) & \( 0.9999 \, / \, 0.9999 \) \\
 & \( 10^{-3} \) & \( 0.9999 \, / \, 0.9999 \) & \( 0.9999 \, / \, 0.9999 \) & \( 0.9999 \, / \, 0.9999 \) & \( 0.9999 \, / \, 0.9999 \) \\
 & \( 10^{-4} \) & \( 0.9999 \, / \, 0.9999 \) & \( 0.9999 \, / \, 0.9999 \) & \( 0.9999 \, / \, 0.9999 \) & \( 0.9999 \, / \, 0.9999 \) \\
 & \( 10^{-5} \) & \( 0.9999 \, / \, 0.9999 \) & \( 0.9999 \, / \, 0.9999 \) & \( 0.9999 \, / \, 0.9999 \) & \( 0.9999 \, / \, 0.9999 \) \\
\midrule
\multirow{4}{*}{\( 0.25 \)}
 & \( 10^{-2} \) & \( 0.9999 \, / \, 0.9999 \) & \( 0.9999 \, / \, 0.9999 \) & \( 0.9999 \, / \, 0.9999 \) & \( 0.9999 \, / \, 0.9985 \) \\
 & \( 10^{-3} \) & \( 0.9999 \, / \, 0.9999 \) & \( 0.9999 \, / \, 0.9999 \) & \( 0.9999 \, / \, 0.9999 \) & \( 0.9999 \, / \, 0.9853 \) \\
 & \( 10^{-4} \) & \( 0.9999 \, / \, 0.9999 \) & \( 0.9999 \, / \, 0.9999 \) & \( 0.9999 \, / \, 0.9999 \) & \( 0.9985 \, / \, 0.8628 \) \\
 & \( 10^{-5} \) & \( 0.9999 \, / \, 0.9999 \) & \( 0.9999 \, / \, 0.9999 \) & \( 0.9999 \, / \, 0.9999 \) & \( 0.9853 \, / \, 0.1500 \) \\
\midrule
\multirow{4}{*}{\( 0.75 \)}
 & \( 10^{-2} \) & \( 0.9999 \, / \, 0.9999 \) & \( 0.9999 \, / \, 0.9999 \) & \( N \le 0 \) & \( N \le 0 \) \\
 & \( 10^{-3} \) & \( 0.9999 \, / \, 0.9999 \) & \( 0.9999 \, / \, 0.9999 \) & \( N \le 0 \) & \( N \le 0 \) \\
 & \( 10^{-4} \) & \( 0.9999 \, / \, 0.9999 \) & \( 0.9999 \, / \, 0.9999 \) & \( N \le 0 \) & \( N \le 0 \) \\
 & \( 10^{-5} \) & \( 0.9999 \, / \, 0.9999 \) & \( 0.9999 \, / \, 0.9999 \) & \( N \le 0 \) & \( N \le 0 \) \\
\bottomrule
\end{tabular}
\end{table}
\clearpage
\section{Conclusions and future work}
This note examined the sampling properties of the AR(1)-jittered lattice \( t_n = n + \xi_n \) at two scales. At the infinite-dimensional scale, Lemmas \ref{lemma:exp_decay} and \ref{lemma:exp_decay_AR1} show that the jittered set attains the critical Nyquist density both in expectation and almost surely, yet Theorem \ref{thm:ar_jitt_notstable} shows that, with probability one, it fails to be a stable sampling set for \( \mathrm{PW}_{[-1/2,1/2]} \). The mechanism behind this failure, adapted from \cite{bass_groechenig}, exploits the ergodicity of the shift on the AR(1) path to locate, almost surely, an arbitrarily long block on which the jitter aligns with the zeros of a fixed bandlimited function. At the finite-dimensional scale relevant to ADC applications, Propositions \ref{prop:prob_lb_sigmamin} and \ref{prop:linearized_prob_lb} give explicit, non-asymptotic conditions on the sample size \( N \) under which the jittered sinc sampling matrix, and its first-order Taylor surrogate, respectively, is well-conditioned with high probability. The Monte Carlo experiments of Section \ref{sec:num_exp} corroborate these guarantees: for \( N = 1000 \) and \( \sigma_\xi = 0.04 \), the smallest singular values of both matrices remain above \( 0.8 \) and the condition numbers below \( 1.5 \) across all tested correlation levels \( \varphi \in [0.1, 0.99] \), with conditioning improving monotonically in \( \varphi \); this dependence on \( \varphi \) is not captured by the bounds, which involve only the marginal variance \( \sigma_\xi^2 \).

Future work will focus on understanding whether ergodicity can again be exploited to show that the autoregressive jittered sinc family also violates the upper frame (Bessel) bound, and whether the autoregressive structure can be leveraged to sharpen the conclusions of Propositions \ref{prop:prob_lb_sigmamin} and \ref{prop:linearized_prob_lb}.

\newpage
\section*{Declaration of generative AI and AI-assisted technologies in the writing process}

During the preparation of this work, the authors occasionally used Anthropic's Claude as interlocutor for discussions of mathematical arguments, exposition, and English grammar. All mathematical content, proofs, results, code, and conclusions are the authors' own, as is responsibility for any errors.

\bibliographystyle{elsarticle-num}
\bibliography{refs}

\end{document}